\documentclass[a4paper,UKenglish,cleveref,autoref,thm-restate]{lipics-v2021}

\graphicspath{{./figures/}}

\title{Efficiently Testing Simon's Congruence}

\author{Pawe\l{}  Gawrychowski}{University of Wroc\l{}aw, Faculty of Mathematics and Computer Science, Poland}{gawry@cs.uni.wroc.pl}{https://orcid.org/0000-0002-6993-5440}{}

\author{Maria Kosche}{G\"ottingen University, Computer Science Department, Germany}{maria.kosche@cs.uni-goettingen.de}{https://orcid.org/0000-0002-2165-2695}{}

\author{Tore Ko\ss }{G\"ottingen University, Computer Science Department, Germany}{tore.koss@cs.uni-goettingen.de}{https://orcid.org/0000-0001-6002-1581}{}

\author{Florin Manea }{G\"ottingen University, Computer Science Department and Campus-Institut Data Science, Germany}{florin.manea@cs.uni-goettingen.de}{https://orcid.org/0000-0001-6094-3324}{}

\author{Stefan Siemer}{G\"ottingen University, Computer Science Department, Germany}{stefan.siemer@cs.uni-goettingen.de}{https://orcid.org/0000-0001-7509-8135}{}

\authorrunning{P. Gawrychowski, M. Kosche, T. Ko\ss, F. Manea, S. Siemer} 

\Copyright{Pawe\l{} Gawrychowski, Maria Kosche, Tore Ko\ss, Florin Manea, Stefan Siemer} 

\ccsdesc[500]{Theory of computation~Formal languages and automata theory}
\ccsdesc[500]{Theory of computation~Design and analysis of algorithms} 

\keywords{Simon's congruence, Subsequence, Scattered factor, Efficient algorithms} 

\category{} 

\relatedversiondetails{STACS 2021}{https://doi.org/10.4230/LIPIcs.STACS.2021.34}

\supplement{}
\funding{The work of the four authors from G\"ottingen was supported by the DFG-grant 389613931 {\em Kombinatorik der Wortmorphismen}.}

\nolinenumbers 

\hideLIPIcs  

\EventEditors{Markus Bl\"{a}ser and Benjamin Monmege}
\EventNoEds{2}
\EventLongTitle{38th International Symposium on Theoretical Aspects of Computer Science (STACS 2021)}
\EventShortTitle{STACS 2021}
\EventAcronym{STACS}
\EventYear{2021}
\EventDate{March 16--19, 2021}
\EventLocation{Saarbr\"{u}cken, Germany (Virtual Conference)}
\EventLogo{}
\SeriesVolume{187}
\ArticleNo{39}

\usepackage{dsfont}
\usepackage{stmaryrd}
\usepackage{tikz}
\usetikzlibrary{positioning}
\usetikzlibrary{arrows}
\usetikzlibrary{decorations.markings}
\usepackage{graphics}
\usepackage{standalone}
\usepackage[linesnumbered,ruled,vlined]{algorithm2e}
\usepackage{amsmath}
\usepackage{amssymb}
\usepackage{mathtools}
\usepackage{mdframed}
\usepackage{todonotes}

\newcommand{\Adag}{a^{\dag}} 
\newcommand{\alp}{\text{alph}}

\newcommand{\len}[1]{|#1|}

\newcommand{\SF}{\ScatFact}

\newcommand{\w}{w}
\newcommand{\ww}{w'}
\newcommand{\SFdagi}{\SF_1(i,a^{\dag})}
\newcommand{\SFdagj}{\SF_1(j,a^{\dag})}

\newcommand{\maxkproblem}{\textsc{MaxSimK}}
\newcommand{\kdecision}{\textsc{SimK}}
\DeclareMathOperator{\splitchar}{\mathtt{split}}
\DeclareMathOperator{\nextpos}{\mathtt{next}}
\DeclareMathOperator{\prevpos}{\mathtt{prev}}

\DeclareMathOperator{\rightpos}{\mathtt{right}}
\DeclareMathOperator{\findNode}{\mathtt{findNode}}
\DeclareMathOperator{\splitNode}{\mathtt{splitNode}}
\DeclareMathOperator{\find}{\mathtt{find}}
\DeclareMathOperator{\union}{\mathtt{union}}
\DeclareMathOperator{\ScatFact}{Subseq}

\theoremstyle{plain}

\theoremstyle{definition}
\newtheorem{problem}{Problem}

\theoremstyle{remark}
\newtheorem{case}{Case}

\makeatletter
\@addtoreset{case}{theorem}
\@addtoreset{case}{lemma}
\makeatother

\begin{document}

\colorlet{lipicsYellow}{white}
\maketitle

\begin{abstract}
Simon's congruence $\sim_k$ is a relation on words defined by Imre Simon in the 1970s and intensely studied since then. This congruence was initially used in connection to piecewise testable languages, but also found many applications in, e.g., learning theory, databases theory, or linguistics. The $\sim_k$-relation is defined as follows: two words are $\sim_k$-congruent if they have the same set of subsequences of length at most $k$. A long standing open problem, stated already by Simon in his initial works on this topic, was to design an algorithm which computes, given two words $s$ and $t$, the largest $k$ for which $s\sim_k t$. We propose the first algorithm solving this problem in linear time $O(|s|+|t|)$ when the input words are over the integer alphabet $\{1,\ldots,|s|+|t|\}$ (or other alphabets which can be sorted in linear time). Our approach can be extended to an optimal algorithm in the case of general alphabets as well. 

To achieve these results, we introduce a novel data-structure, called Simon-Tree, which allows us to construct a natural representation of the equivalence classes induced by $\sim_k$ on the set of suffixes of a word, for all $k\geq 1$. We show that such a tree can be constructed for an input word in linear time. Then, when working with two words $s$ and $t$, we compute their respective Simon-Trees and efficiently build a correspondence between the nodes of these trees. This correspondence, which can also be constructed in linear time $O(|s|+|t|)$, allows us to retrieve the largest $k$ for which $s\sim_k t$.
%
\end{abstract}

\newpage
\section{Introduction}
A subsequence of a word $w$
(also called scattered factor or subword, especially in automata and language theory)
is a word $u$ such that there exist (possibly empty) words $v_1, \ldots, v_{\ell+1}, $ $u_1, \ldots, 
u_\ell$ with $u = u_1 \ldots u_n$ and $w = v_1 u_1 v_2 u_2 \ldots v_\ell u_{\ell} v_{\ell+1}$.
Intuitively, the subsequences of a word $w$ are exactly those words
obtained by deleting some of the letters of $w$,
so, in a sense, they can be seen as lossy-representations of the word $w$.
Accordingly, subsequences may be a natural mathematical model for situations
where one has to deal with input strings with missing or erroneous symbols sequencing,
such as processing DNA data or digital signals~\cite{sankoff}.
Due to this very simple and intuitive definition
as well as the apparently large potential for applications,
there is a high interest in understanding the fundamental properties
that can be derived in connection to the sets of subsequences of words.
This is reflected in the consistent literature developed around this topic.
J. Sakarovitch and I. Simon in \cite[Chapter 6]{Loth97} overview some of the most important combinatorial and language theoretic properties of sets of subsequences.
The theory of subsequences was further developed in various directions, such as 
combinatorics on words, automata theory, formal verification, or string algorithms.
For instance, subword histories and Parikh matrices
(see, e.g.,~\cite{Mat04,Salomaa05,Seki12})
are algebraic structures in which the number of specific subsequences occurring in a word
are stored and used to define combinatorial properties of words.
Strongly related,
the binomial complexity of words is a measure of the multiset of subsequences that occur in a word,
where each occurrence of such a factor is considered as an element of the respective multiset;
combinatorial and algorithmic results related to this topic are obtained in,
e.g., \cite{RigoS15,FreydenbergerGK15,LeroyRS17a,Rigo19}, and the references therein.
Further, in~\cite{Zetzsche16,HalfonSZ17,KuskeZ19} various logic-theories were developed,
starting from the subsequence-relation,
and analysed mostly with automata theory and formal verification tools.
Last, but not least,
many classical problems in the area of string algorithms are related to subsequences:
the longest common subsequence
and the shortest common supersequence problems \cite{Maier:1978,DBLP:journals/tcs/Baeza-Yates91,DBLP:conf/fsttcs/BringmannC18,BringmannK18},
or the string-to-string correction problem \cite{Wagner:1974}.
Several algorithmic problems connected to subsequence-combinatorics are approached and (partially) solved in \cite{ElzingaRW08}.\looseness=-1

One particularly interesting notion related to subsequences
was introduced by Simon in~\cite{Simon72}.
He defined the relation $\sim_k$
(called now Simon's congruence)
as follows.
Two words are $\sim_k$-congruent
if they have the same set of subsequences of length at most $k$.
In~\cite{Simon72}, as well as in \cite[Chapter 6]{Loth97},
many fundamental properties
(mainly of combinatorial nature)
of $\sim_k$ are discussed;
this line of research was continued in,
e.g., \cite{KarandikarKS15,CSLKarandikarS,journals/lmcs/KarandikarS19,dlt2019,dlt2020}
where the focus was on the properties of some special classes of this equivalence.
From an algorithmic point of view,
a natural decision problem and its optimisation variant stand out:

\begin{problem}\kdecision\em{: }
	Given two words $s$ and $t$ over an alphabet $\Sigma$, with $|s|=n$ and $|t|=n'$, with $n\geq n'$,
	and a natural number $k$,
	decide whether $s \sim_k t$.
\end{problem}

\begin{problem}\maxkproblem\em{: }
	Given two words $s$ and $t$ over an alphabet $\Sigma$, with $|s|=n$ and $|t|=n'$, with $n\geq n'$,
	find the maximum $k$ for which $s \sim_k t$.
\end{problem}

The problems above were usually considered assuming
(although not always explicitly)
the Word RAM model with words of logarithmic size.
This is a standard computational model in algorithm design
in which, for an input of size $n$, the memory consists of memory-words consisting of $\Theta(\log n)$ bits.
Basic operations
(including arithmetic and bitwise Boolean operations)
on memory-words take constant time,
and any memory-word can be accessed in constant time.
In this model,
the two input words are just sequences of integers,
each integer stored in a single memory-word.
Without losing generality,
we can assume the alphabet to be simply $\{1,\ldots,n+n'\}$
by sorting and renaming the letters occurring in the input in linear time.
For a detailed discussion on the computational model, see \cref{sec:CompModel}.\looseness=-1

Due to the central role played by the $\sim_k$-congruence in the study of piecewise testable languages, as well as in the many other areas mentioned above, 
both problems \kdecision\ and \maxkproblem\ were considered highly interesting and studied thoroughly in the literature.

In particular, Hebrard \cite{TCS::Hebrard1991} presents \maxkproblem\ as computing a similarity measure between strings
and mentions a solution of Simon \cite{Simon_unpublished} for \maxkproblem\
which runs in $O(|\Sigma|nn')$
(the same solution is mentioned in \cite{garelCPM}).
He goes on and improves this
(see \cite{TCS::Hebrard1991})
in the case when $\Sigma$ is a binary alphabet:
given two bitstrings $s$ and $t$,
one can find the maximum $k$ for which $s \sim_k t$ in linear time.
The problem of finding optimal algorithms for \maxkproblem,
or even \kdecision,
for general alphabets was left open in \cite{Simon_unpublished,TCS::Hebrard1991}
as the methods used in the latter paper for binary strings did not seem to scale up.
In \cite{garelCPM}, Garel considers \maxkproblem\ and presents an algorithm based on finite automata,
running in $O(|\Sigma|n)$,
which computes all {\em distinguishing words} $u$ of minimum length,
i.e., words which are factors of only one of the words $s$ and $t$ from the problem's statement.
Several further improvements on the aforementioned results were reported in \cite{DBLP:journals/jda/CrochemoreMT03,DBLP:conf/wia/Tronicek02}.
Finally, in an extended abstract from 2003 \cite{SimonWords},
Simon presented another algorithm based on finite automata solving \maxkproblem\
which runs in $O(|\Sigma|n)$,
and he conjectures that it can be implemented in $O(|\Sigma|+n)$.
Unfortunately, the last claim is only vaguely and insufficiently substantiated,
and obtaining an algorithm with the claimed complexity seems to be non-trivial.
Although Simon announced that a detailed description of this algorithm will follow shortly,
we were not able to find it in the literature.\looseness=-1

In \cite{KufMFCS}, a new approach to efficiently solving \kdecision\ was introduced.
This idea was to compute,
for the two given words $s$ and $t$ and the given number $k$,
their shortlex forms:
the words which have the same set of subsequences of length at most $k$ as $s$ and $t$, respectively,
and are also lexicographically smallest among all words with the respective property.
Clearly, $s\sim_k t$ if and only if the shortlex forms of $s$ and $t$ for $k$ coincide.
The shortlex form of a word $s$ of length $n$ over $\Sigma$ was computed in $O(|\Sigma|n)$ time,
so \kdecision\ was also solved in $O(|\Sigma| n)$.
A more efficient implementation of the ideas introduced in~\cite{KufMFCS} was presented in \cite{dlt2020}:
the shortlex form of a word of length $n$ over $\Sigma$ can be computed in linear time $O(n)$,
so \kdecision\ can be solved in optimal linear time.
By binary searching for the smallest $k$, this gives an $O(n\log n)$ time solution for \maxkproblem.
This brings up the challenge of designing an optimal linear-time algorithm for non-binary alphabets.\looseness=-1

\vspace{-7pt}
\subparagraph*{Our results.}
In this paper we confirm Simon's claim from 2003 \cite{SimonWords}. We present a complete algorithm solving \maxkproblem\ in linear time
on Word RAM with words of size $\Theta(\log n)$.
This closes the problem of finding an optimal algorithm for \maxkproblem.
Our approach is not based on finite automata (as the one suggested by Simon),
nor on the ideas from \cite{KufMFCS,dlt2020}.
Instead, it works as follows.~Firstly,~looking~at~a~single word,
we partition the respective word into $k$-blocks:
contiguous intervals of positions inside the word,
such that all suffixes of the word inside the same block have exactly the same subsequences of length at most $k$
(i.e., they are $\sim_k$-equivalent).
Since the partition in~$(k+1)$-blocks refines the partition in $k$-blocks,
one can introduce the {\em Simon-Tree} associated~to~a~word:
its nodes are the $k$-blocks
(for $k$ from $1$ to at most $n$),
and each node on level $k$ has as children exactly the $(k+1)$-blocks in which it is partitioned.
We first show how to compute efficiently the Simon-Tree of a word.
Then, to solve \maxkproblem,
we show that one can maintain in linear time a connection between the nodes on the same levels of the Simon-Trees associated to the two input words.
More precisely, for all $\ell$,
we connect two nodes on level $\ell$ of the two trees if the suffixes starting in those blocks,
in their respective words,
have exactly the same subsequences of length at most $\ell$.
It follows that the value $k$ required in \maxkproblem\ is the lowest level of the trees
on which the blocks containing the first position of the respective input words are connected.
Using the Simon-Trees of the two words and the connection between their nodes,
we can also compute in linear time a distinguishing word of minimal length for $s$ and $t$.
Achieving the desired complexities is based on a series of combinatorial properties of the Simon-Trees,
as well as on a rather involved data structures toolbox.\looseness=-1

Our paper is structured as follows:
we firstly introduce the basic combinatorial and data structures notions in \cref{sec:preliminaries},
then we show how Simon-Trees are constructed efficiently in \cref{sec:simon-tree},
and,~finally, we show how \maxkproblem\ is solved by connecting the Simon-Trees of the two input words in \cref{sec:connecting-two-simon-trees,sec:s-connection-algo}. 
We end this paper with \cref{sec:conclusions} containing a series of concluding remarks, extensions, and further research questions.
A discussion on how our results can be extended to an optimal algorithm for \maxkproblem\ in the case of input words over general alphabets is also given in \cref{sec:CompModel}. \looseness=-1

\section{Preliminaries}
\label{sec:preliminaries}
Let $\mathds{N}$ be the set of natural numbers, including $0$.
An alphabet $\Sigma$ is a nonempty finite set of symbols called {\em letters}.
A {\em word} is a finite sequence of letters from $\Sigma$,
thus an element of the free monoid $\Sigma^{\ast}$.
Let $\Sigma^+=\Sigma^{\ast}\backslash\{\varepsilon\}$,
where $\varepsilon$ is the empty word.
The {\em length} of a word $w\in\Sigma^{\ast}$ is denoted by $|w|$.
The $i^{th}$ letter of $w\in\Sigma^{\ast}$ is denoted by  $w[i]$, for $i\in [1:\len w]$.
For $m, n\in\mathds{N}$,
we let $[m:n]=\{m,m+1,\ldots,n\}$ and $\w[m:n] = \w[m]\w[m+1]\ldots \w[n]$.

A word $u\in\Sigma^{\ast}$ is a {\em factor} of $w\in\Sigma^{\ast}$
if $w=xuy$ for some $x,y\in\Sigma^{\ast}$.
If $x=\varepsilon$ (resp. $y=\varepsilon$),
$u$ is called a  {\em prefix} (resp. {\em suffix} of $w$). 
For some $x\in \Sigma$ and $w\in \Sigma^*$,
let $|w|_{x}=|\{i\in[1:|w|] \mid w[i]=x\}|$ and $\alp(w)$ $= \{x \in \Sigma \mid |w|_x > 0 \}$ for $w\in\Sigma^{\ast}$;
in other words, $\alp(\w)$ denotes the smallest subset $S\subset \Sigma$ such that $\w\in S^*$.

\begin{definition}
	We call $u$ a subsequence of length $k$ of $\w$,
	where $|\w|=n$,
	if there exist positions $1\leq i_1<i_2<\ldots< i_k \leq n$,
	such that $u = \w[i_1]w[i_2]\cdots \w[i_k]$.
	Let $\SF_{\leq k}(i,\w)$ denote the set of subsequences of length at most $k$ of $\w[i:n]$.
	Accordingly, the set of subsequences of length at most $k$ of the entire word $\w$
	will be denoted by $\SF_{\leq k}(1,\w)$.
\end{definition}

Equivalently, $u= u_1 \ldots u_\ell $ is a subsequence of $w$
if there exist 
$v_1, \ldots,  v_{\ell+1}\in\Sigma^{\ast}$
such that $w = v_1 u_1 \ldots v_\ell u_\ell v_{\ell+1}$.
For $k\in\mathds{N}$, $\SF_{\leq k}(1,\w)$ is called the full $k$-spectrum of~$\w$.

\begin{definition}[Simon's Congruence]\label{SC}
	(i) Let $\w,\ww \in \Sigma^*$.
	We say that $\w$ and $\ww$ are equivalent under Simon's congruence $\sim_k$
	(or, alternatively, that $\w$ and $\ww$ are $k$-equivalent)
	if the set of subsequences of length at most $k$ of $\w$
	equals the set of subsequences of length at most $k$ of $\ww$,
	i.e., $\SF_{\leq k}(1,\w) = \SF_{\leq k}(1,\ww)$.
	
	(ii) Let $i,j \in [1:|w|]$.
	We define $i\sim_k j$ (w.r.t. $\w$)
	if $\w[i:n] \sim_k \w[j:n]$,
	and we say that the positions $i$ and $j$ are $k$-equivalent.
	
	(iii) A word $u$ of length $k$ distinguishes $\w$ and $\ww$ w.r.t. $\sim_k$
	if $u$ occurs in exactly one of the sets $\SF_{\leq k}(1,\w)$ and $\SF_{\leq k}(1,\ww)$.
\end{definition}

Following the discussion from the introduction,
for our algorithmic results we assume the Word RAM model with words of size $\Theta(\log n)$.

At this point, we also want to recall two data structures which play an important role in our results. These are the interval split-find and interval union-find data structures.
Their formal definition is given in the following. 

\begin{definition}[Interval split-find and interval union-find data structures]\label{SplitUnionFind}
	Let $V=[1:n]$ and $S$ a set with $S \subseteq V$.
	The elements of $S =\{ s_1, \ldots , s_p \}$ are called borders
	and are ordered $0 = s_0 < s_1 < \ldots < s_p < s_{p+1} = n+1$
	where $s_0$ and $s_{p+1}$ are generic borders.
	For each border $s_i$,
	we define $V(s_i) =[s_{i-1} +1: s_i]$ as an induced interval.
	Now, $P(S) \coloneqq \lbrace V(s_i)~|~ s_i \in S \rbrace$ gives an ordered partition of the interval $V$.
	\begin{itemize}
		\item The interval split-find structure maintains the partition $P(S)$ under the operations:
		\begin{itemize}
			\item For $u \in V$,
			$\find(u)$ returns $s_i \in S \cup \lbrace n+1 \rbrace$
			such that $u \in V(s_i)$.
			In other words, all elements in the interval $V(s_i)$ have the representative $s_i$.
			\item For $u \in V \setminus S$,
			$\splitchar(u)$ updates the partition $P(S)$ to $P(S \cup \{ u\} )$.
			That is, we find the interval $V(s_i)$ containing $u$,
			split it into the interval containing elements $\leq u$ and the interval of elements $>u$,
			and update the partition so that further $\find$ and $\splitchar$ operations can be performed.
		\end{itemize}
		\item The interval union-find structure maintains the partition $P(S)$ under the operations:
		\begin{itemize}
			\item For $u \in V$,
			$\find(u)$ returns $s_i \in S \cup \lbrace n+1 \rbrace$
			such that $u \in V(s_i)$.
			\item For $u \in S$,
			$\union(u)$ updates the partition $P(S)$ to $P(S \setminus \{ u\} )$.
			That is, if $u=s_1$,
			then we replace the intervals $V(s_i)$ and $V(s_{i+1})$ by the single interval $[s_{i-1}+1:s_{i+1}]$
			and update the partition
			so that further $\find$ and $\union$ operations can be performed.
		\end{itemize}
	\end{itemize}
\end{definition}

Rather informally, in the union-find (respectively, split-find) structure we maintain a 
partition of an interval (also called universe) $V=[1:n]$ in sub-intervals, under two operations: 
$\union$ of adjacent intervals (respectively, $\splitchar$ an interval in two sub-intervals around an element of the interval), and $\find$ the representative of the interval containing a given value. 
In our algorithms, when using these structures, we usually describe the intervals stored initially in the structure,
and then the $\union$s (respectively, the $\splitchar$s) which are made, as well as the $\find$ operations, 
without going into the formalism behind these operations. In usual implementations of these structures the representative of each interval, which is returned by $\find$, is its maximum;
we can easily enhance the data structures so that the $\find$ operation returns both borders of the interval containing the searched value. 
The following lemma was shown in~\cite{GabowT85,imaiUF}.

\begin{lemma}\label{union-find}
	One can implement the interval split-find (respectively, union-find) data structures,
	such that, the initialisation of the structures followed by a sequence of $m\in O(n)$ $\splitchar$ (respectively, $\union$) and $\find$ operations can be executed in $\mathcal{O}(n)$ time and space.
\end{lemma}

\section{Constructing the Simon-Tree of a word}
\label{sec:simon-tree}
In this section, we introduce a new data structure, 
which is fundamental to our approach -- the \emph{Simon-Tree}.
The Simon-Tree is used as a representation for the equivalence classes in a word,
which are explained in \cref{sec:equivalence-classes}.
The definition of Simon-Trees is then given in \cref{sec:simon-tree-definition},
and the construction is described in \cref{sec:simon-tree-construction}.

\subsection{Equivalence classes of a Word}
\label{sec:equivalence-classes}
In this section,
we develop a method to efficiently partition the positions of a given word $w$, of length $n$, 
into equivalence classes w.r.t. $\sim_k$,
such that all suffixes starting with positions of the same class have the same set of subsequences of length at most $k$.
As in this section we only deal with one input word $w$ of length $n$,
we will sometimes omit the reference to this word in our notation:
e.g. $\SF_k(i) = \SF_k(i,\w)$;
in the case of such omissions,
the reader may safely assume that we are referring to the aforementioned input word.

Firstly, we will examine the equivalence classes that each congruence relation $\sim_k$ induces on the set of suffixes of $\w$ for all $k$.
Let $1 \leq i<j\leq n$, then $\w[j:n]$ is a suffix of $\w[i:n]$,
hence $\SF_k(i) \supseteq \SF_k(j)$ holds for all $k\in\mathds{N}$.
For any $l\in [i:j]$ we obtain $\SF_k(i) \supseteq \SF_k(l) \supseteq \SF_k(j)$.
If we additionally let $i\sim_k j$,
then the sets of subsequences corresponding to $i$ and $j$ respectively are equal,
so $\SF_k(i)=\SF_k(l)=\SF_k(j)$ and $i\sim_k l\sim_k j$.
Hence, the equivalence classes of the set of suffixes of $w$ w.r.t. $\sim_k$
correspond to sets of consecutive indices (i.e., intervals) in $[1:\len w]$,
namely the starting positions of the suffixes in each class.
We call these classes $k$-\textit{blocks}.\looseness=-1

A $k$-block consisting only of a single position
(i.e., it is a {\em singleton-$k$-block}),
remains an $\ell$-block for all $\ell>k$.
For a $k$-block $b=[m_b:n_b]$,
$m_b$ is its starting position and $n_b$ its ending position.
For the ending position of a $k-$block we also use the following definition.\looseness=-1

\begin{definition}
	For some $k>0$, if $i\sim_{k-1} i+1$ and $i\nsim_{k} i+1$,
	then we will say that $i$ \textit{splits} its $(k-1)$-block or that $i$ is a $k$-\textit{splitting position}.
\end{definition}

If $a=[m_a:n_a]$ is a $k$-block and $b=[m_b:n_b]$ is a $(k+1)$-block with $m_a\leq m_b\leq n_b\leq n_a$,
then we say that $b$ is a $(k+1)-$block in $a$ (alternatively, of $a$).

Since, for $1 \leq i,j \leq n$, $i\sim_k j$ holds if $i\sim_{k+1}j$,
the relation $\sim_{k+1}$ is a refinement of $\sim_k$.
In our setting, this means that the $(k+1)-$blocks of $w$ are obtained by partitioning the $k$-blocks of $w$ into subintervals.
To obtain a partition of the positions of $w$ into equivalence classes and the corresponding blocks,
we can use this refinement property.
We get the following inductive procedure.\looseness=-1

\begin{description}
	\item[$0$-blocks] For any $i$, with $1 \leq i \leq n$, we have $\SF_0(i) = \{\varepsilon\}$.
	Thus, we refer to $[1:n]$ as the $0$-\textit{block} of $\w$.
	Note, however, that position $n$ cannot be referred to as a $0$-splitting position.\looseness=-1
	\item[$(k+1)$-blocks] For $k \geq 0$ and a $k$-block $a = [m_a : n_a]$ in $\w$ with $\len a \geq 2$,
	we can find the $(k+1)$-splitting positions inside of $a$.
	Except for the case when $a$ is a $0$-block,
	position $n_a$ marks a $k$-splitting position.
	So, if $k > 0$, we slice off position $n_a$
	and obtain a truncated block $\Adag$; 
	if $k = 0$, then $\Adag = a$.
	Going from right to left through $\Adag$,
	the position of every character we encounter for the first time
	(so, which we have not seen before in this traversal of $\Adag $) is a splitting position of $\Adag$.
	Consequently, those splitting positions and $n_a$ (only if $k>0$) will split $a$ into $(k+1)$-blocks.
	The correctness of this approach follows from \cref{lem:split}.
\end{description}

\begin{restatable}{lemma}{restateSplit}\label{lem:split}
	Let $a=[m_a,n_a]$ be a $k$-block with $1 \leq m_a < n_a\leq n$. Let  $\Adag = a$ for $k=0$ and $\Adag = [m_a : n_a - 1]$ for $k > 0$. 
	Then the following holds for all $i,j \in \Adag$: 
	\begin{enumerate}[(i)]
		\item if $k>0$ then $i \nsim_{k+1} n_a$;
		\item $i\sim_{k+1} j$ if and only if $\SFdagi = \SFdagj$.
	\end{enumerate}
\end{restatable}
\begin{proof}
	\begin{enumerate}[(i)]
		\item Firstly, let $i\in\Adag$ and $k>0$, and assume $i\sim_{k+1} n_a$.
		We choose an arbitrary $v'\in \SF_{k}(i+1)$
		and let $v = \w[i]v'$.
		Now, $v\in \SF_{k+1}(i)$ hence $v\in\SF_{k+1}(n_a)$.
		But then $v'\in\SF_k(n_a+1)$ and, as $v'$ was arbitrary chosen, then $i+1 \sim_k n_a+1$,
		which contradicts the definition of $a$ being a $k$-block. Note that for $k=0$ we would have that $n_a+1>n$ so our reasoning does not hold.
		\item For $k=0$ the statement follows directly from \cref{SC}.
		So let $k > 0$ and w.l.o.g. $i<j$.
		For the first implication let $i\sim_{k+1} j$ and assume $\SFdagi \neq \SFdagj$.
		Then $\SFdagj$ is properly contained in $\SFdagi$,
		hence we can choose $c\in\SFdagi \setminus \SFdagj$ and find a left-most position $\ell>j$ such that $w[l] = c$ (note that $\ell\geq n_a$).
		For an arbitrary $v'\in \SF_k(j)\setminus\SF_k(n_a+1)$ we have $cv'\in \SF_{k+1}(i)=\SF_{k+1}(j)$.
		But then $cv'\in\SF_{k+1}(\ell)$ and therefore $v'\in\SF_k(\ell+1)\subset\SF_k(n_a+1)$, a contradiction.
		
		Lastly let $\SFdagi = \SFdagj$ and assume $i\nsim_{k+1} j$.
		We can find $v\in\SF_{k+1}(i)\setminus\SF_{k+1}(j)$.
		Necessarily, we have $v[1]\in \alp(w[i:j-1])$ hence $v[1]\in\SFdagi=\SFdagj$.
		Therefore we can find a (unique left-most) position $\ell\in [j:n_a-1]$ such that $w[\ell] = v[1]$.
		Now $v[1]^{-1}v\in\SF_k(\ell+1)$, because $\ell+1\sim_k i$.
		This implies $v\in\SF_{k+1}(\ell)\subset\SF_{k+1}(j)$, a contradiction.
	\end{enumerate}
\end{proof}

\subsection{Simon-Tree definition}
\label{sec:simon-tree-definition}
Before introducing the Simon-Tree,
we recall some basic notions.
An ordered rooted tree is a rooted tree
which has a specified order for the subtrees of a node.
We say that the depth of a node is the length of the unique simple path from the root to that node.
Generally, the nodes with smaller depth are said to be {\em higher}
(the root is the highest node with depth $0$),
while the nodes with greater depth are {\em lower} in the tree.

We can now define a new data structure called \emph{Simon-Tree}.
The Simon-Tree of a word $w$ gives us a hierarchical representation of the equivalence classes inside of $w$.
While an example of a Simon-Tree can be seen in \cref{fig:simon-tree},
the formal definition of a Simon-Tree is as follows.

\begin{figure}
	\centering
	\includegraphics[width=0.7\linewidth]{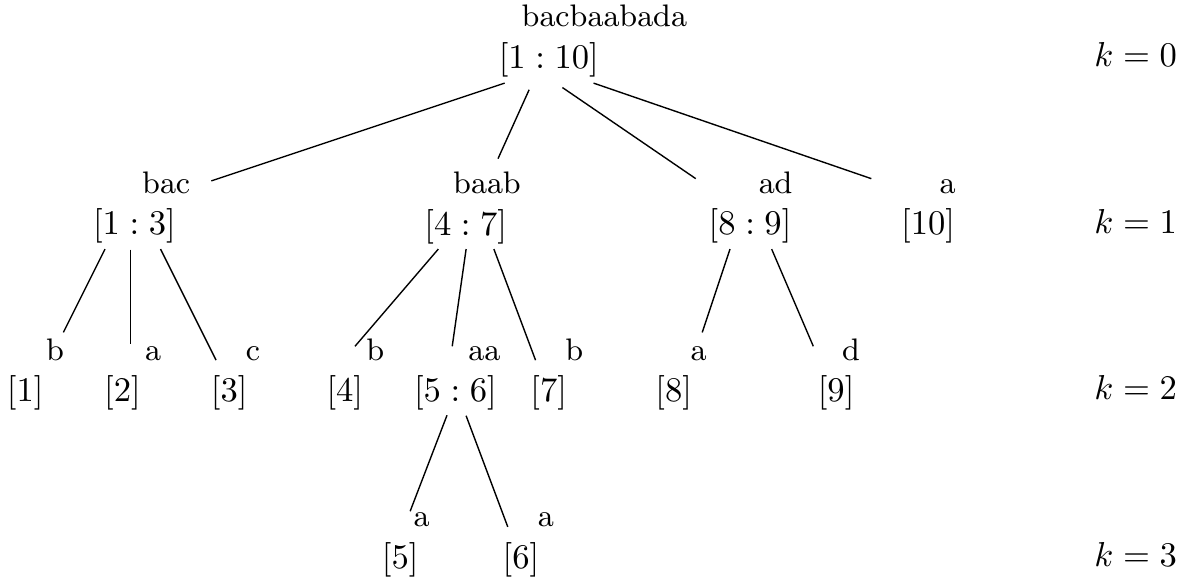}
	\caption{Simon-Tree of the word $bacbaabada$. Above each block $[i:j]$ we wrote the word $w[i:j]$.}
	\label{fig:simon-tree}
\end{figure}

\begin{definition}
	The {\em Simon-Tree} $T_w$ associated to the word $w$, with $|w|=n$,
	is an ordered rooted tree.
	The nodes of depth $k$ represent $k-$blocks of $w$, for $0\leq k\leq n$,
	and are defined recursively.\looseness=-1
	\begin{itemize}
		\item The root corresponds to the $0$-block of the word $w$,
		i.e., the interval $[1:n]$.
		\item For $k>1$ and for a node $a$ of depth $k-1$,
		which represents a $(k-1)$-block $[i:j]$ with $i<j$,
		the children of $a$ are exactly the blocks of the partition of $[i:j]$ in $k$-blocks,
		ordered decreasingly (right-to-left) by their starting position.
		\item For $k>1$,
		each node of depth $k-1$ which represents a singleton-$(k-1)$-block is a leaf.
	\end{itemize}
\end{definition}
	
The nodes of depth $k$ in a tree $T_w$ are called
{\em explicit $k$-nodes (or simply $k$-nodes)};
by abuse of notation,
we identify each $k$-node by the $k$-block it represents.

With respect to their starting positions in the word,
we number the children nodes (which are blocks) of a node $b$ from right to left.
That is, the $i^{th}$ child of $b$ is the $i^{th}$ block of the partition of $a$,
counted from right to left.
The singleton-$j$-blocks, for $j<k$, are also $k$-blocks,
but they do not appear explicitly as nodes of depth $k$ in the tree $T_w$.
We will say that they are {\em implicit $k$-nodes}.
In other words, an explicit singleton-$j$-node is an implicit $k$-node, for all $k>j$,
and the only child of a $k$-node $[i:i]$ is the implicit $(k+1)$-node $[i:i]$.
The nodes of depth $k$ in the Simon-Tree $T_w$ do not necessarily comprise all the $k$-blocks of $w$,
but they contain explicitly exactly those $k$-blocks of $w$
that were obtained by non-trivially splitting a $(k-1)$-block of $w$ which was not a singleton.

\subsection{Simon-Tree construction}
\label{sec:simon-tree-construction}

We are interested in constructing the Simon-Tree $T_\w$ associated to a word $\w$, with $|\w|=n$,
in linear time.
In this section we give a description of the construction algorithm and its analysis.
The corresponding pseudocode can be seen in \cref{alg:simontree,alg:findnode,alg:splitnode}.

For the algorithms, we use the array $X$ of size $n$ which holds,
for a given position $i$,
the next position of $w[i]$ in the word $w[i+1:n]$.
We formally define this with
$X[i] = \min \lbrace j \mid \w[j] = \w[i], j>i \rbrace$,
while we assume $X[i] = \infty$
if $\w[i] \notin \alp(w[i+1:n])$.
The array can be calculated in $\mathcal{O}(n)$ time and space.

\begin{restatable}{lemma}{restateArrayX}\label{arrayX}
	Given a word $\w$ with length $\len \w = n$.
	The array $X$ can be calculated in $\mathcal{O}(n)$ time and space for the entire word. 
\end{restatable}
\begin{proof}
	The word $\w$ needs to be traversed only once from right to left by maintaining an array of size $| \Sigma |$ with the last occurrence of each character.
	Since $| \Sigma | \leq n$, the results follow immediately.
\end{proof}

As an example consider now the word $\w = \mathtt{bacbaabada}$.
The array $X$ is then depicted as follows.

\begin{table}[h!]
	\centering
	\begin{tabular}{|c|c|c|c|c|c|c|c|c|c|c|c|c}
		\hline
		$i$ & 1 & 2 & 3 & 4 & 5 & 6 & 7 & 8 & 9 & 10\\
		\hline
		$\w[i]$ & $\mathtt{b}$ & $\mathtt{a}$ & $\mathtt{c}$ & $\mathtt{b}$ & $\mathtt{a}$ & $\mathtt{a}$ & $\mathtt{b}$ & $\mathtt{a}$ & $\mathtt{d}$ & $\mathtt{a}$\\
		\hline
		$X[i]$ & 4 & 5 & $\infty$ & 7 & 6 & 8 & $\infty$ & 10 & $\infty$ & $\infty$\\
		\hline
	\end{tabular}
\end{table}

When applying our algorithm to $w$,
we get the tree shown in \cref{fig:simon-tree},
where we represent each node with the block $[i:j]$ it represents accompanied by the word $w[i:j]$.

\SetKwFunction{KwSplitNode}{splitNode}
\SetKwFunction{KwFindNode}{findNode}
\begin{algorithm}[!htb]
	\SetAlgoLined
	\KwIn{Word $w$ with $|w| = n$}
	\KwResult{Simon-Tree~$T_w$}
	$w' \gets w\$ $\;
	Let $T$ be the tree with the root associated to the block $[?:n+1]$ of $w'$\;
	Let $p$ be a pointer to the root of $T$\;
	Compute the array $X[i]$\;
	\BlankLine
	
	\For{i=n \KwTo 1}{
		$a \gets$ \KwFindNode{$i$,$T$,$p$}\;
		$(T,p) \gets$ \KwSplitNode{$i$,$T$,$a$}\;
	}
	\BlankLine
	
	Set starting position for all blocks from leftmost branch including the root to $1$\;
	Remove $\$$-letter from tree: Remove the node associated to $[n+1:n+1]$ from $T$ and set all right ends $r$ of blocks on the rightmost branch to $n$\;
	\Return T \;
	
	\caption{Building the Simon-Tree $T_w$ for a word $w$}
	\label{alg:simontree}
\end{algorithm}
\SetKw{KwAnd}{and}
\begin{algorithm}[!htb]
	\SetAlgoLined
	\KwIn{Position $i$ in $w$, Simon-Tree $T_{w[i+1:n]\$}$, Pointer to leftmost leaf $a$ of $T_{w[i+1:n]\$}$}
	\KwResult{Pointer to node on leftmost branch of $T_{w[i+1:n]\$}$}
	\BlankLine
	
	\While{$a$ is not the root of $T_{w[i+1:n]\$}$}{
		$r \gets$ ending position of the block represented by $a$\;
		$r_p \gets$ ending position of block represented by $parent(a)$\;
		\BlankLine
		
		\uIf{$X[i] \geq r$ \KwAnd $X[i] < r_p$}{
			\Return $a$\;
		} \Else{
			Close the block represented by $a$: Set its starting position to $i+1$\;
			$a \gets parent(a)$\;
		}
	}
	\BlankLine
	
	\Return $a$\;
	
	\caption{$\findNode$}
	\label{alg:findnode}
\end{algorithm}
\begin{algorithm}[!htb]
	\SetAlgoLined
	\KwIn{Position $i$, Simon-Tree $T_{w[i+1:n]\$}$, Pointer to node $a$ on leftmost branch of $T_{w[i+1:n]\$}$}
	\KwResult{Simon-Tree $T_{w[i:n]\$}$, Pointer to leftmost leaf of $T_{w[i:n]\$}$}
	$T \gets T_{w[i+1:n]\$}$ \;
	\uIf{$a$ is the leftmost leaf of $T$}{
		\tcp{node $a$ represents the open block $[?:i+1]$}
		Add a child to $a$ associated to the now completed block $[i+1:i+1]$\;
		Add a leftmost child $b$ to $a$ associated to the open block $[?:i]$\;
	} \Else{
		Add a leftmost child $b$ to $a$ associated to the open block $[?:i]$\;
	}
	\BlankLine
	
	\Return ($T$, pointer to $b$)
	
	\caption{$\splitNode$}
	\label{alg:splitnode}
\end{algorithm}

\subparagraph*{Algorithm description.}
\label{sec:algorithm-description}
In general,
we consider the individual letters of the word $w$ from right to left.
After considering $\w[i]$,
the tree we constructed so far corresponds to the Simon-Tree of the suffix $w[i:n]$.
By traversing the word from right to left,
we also construct the Simon-Tree in a right-to-left manner.
Accordingly, it holds that at each time step only the nodes on the leftmost branch of the tree are possible to be enhanced.
This is because for the tree of the word $w[i+1:n]$,
prepending a new letter to the word $w[i+1:n]$ can only affect the leftmost node/block on each level of the tree,
as the nodes of level $k$ store the $k$-blocks,
and, accordingly, build a
(possibly intermittent, if we only consider the explicit nodes)
partition of the word $w$ into non-overlapping intervals, for all $k$,
while the nodes of one level are ordered with regard to their position in the word.

This means that a newly considered position of our word
can be only added to a node on the leftmost branch of the tree
that was constructed so far during the application of the algorithm.
Therefore, we call the nodes on the leftmost branch \emph{open blocks}.
These open blocks are not complete
and have a yet unknown starting position.
We use $[?:i]$ to denote the open block with unknown starting position
and ending position $i$.
For the nodes on the leftmost branch,
we only store the ending position (or splitting position)
of their represented block.
For all nodes that are not on the leftmost branch in the tree,
we store both starting and ending position of their represented block.

In the beginning of the construction algorithm,
we append the letter $\$$ at the end of $w$ to ensure that all the positions of $w$ are treated in a uniform way.
More precisely, the usage of the $\$ $-letter allows us to uniformly find the splitting points in a block according to case (ii) of \cref{lem:split} only. 
That is, by adding the letter $\$$ at the end,
we avoid position $n$ as being falsely recognized as a $0$-splitting position since it is the ending position of the $0$-block $[1:n]$ of $\w$.
As seen in \cref{alg:simontree},
we define $w' = w \$$
and start the algorithm with the tree that only has one node, the root,
representing the open block $[?:n+1]$.

When considering a new position $i$ of the word,
and, essentially, inserting it into the current tree,
we want to find the correct tree level
where position $i$ would mark the splitting of a new block
or a new node, respectively.
According to \cref{alg:findnode}, by starting at the leftmost leaf
(which is the node associated to the open block $[?:i+1]$)
and going up the leftmost branch of the current tree,
we look for the first node where the character $w[i]$ occurs on a non-ending position.
Let this be node $A$ of depth $k$,
representing an open $k$-block.
Node $A$ cannot be a leaf
since leafs only represent singleton-blocks, consist therefore only of one position,
and $w[i]$ could not occur on a non-ending position.
Let the leftmost child of $A$ be the node $a$ of depth $k+1$.
By utilizing \cref{lem:split},
we get the information
that position $i$ is a $(k+2)$-splitting position in $a$,
and consequently,
our new block with ending position $i$ is mapped to level $(k+2)$ of the Simon-Tree.
Following \cref{alg:splitnode}, we then insert the new block $[?:i]$ in the respective level of the Simon-Tree as a leftmost child of node $a$.

All nodes we traversed from leftmost leaf up to node $a$
represent $l$-blocks with $l \geq (k+2)$.
These blocks are closed during the process of finding the correct position as seen in \cref{alg:findnode}.
Since $i$ is a $(k+2)$-splitting position
we set the starting position for all open $l$-blocks, with $l \geq k+2$,
to $i+1$.\looseness=-1

It remains only to mention the special case,
where we do not find an occurrence of $w[i]$ on our traversal from leftmost leaf to the root.
In this case,
the letter did not appear yet in the word.
It therefore marks a $1$-splitting position
and as per \cref{alg:findnode},
we return the tree root,
to which the block $[?:i]$ is then added as a leftmost child as per \cref{alg:splitnode}.\looseness=-1

\subparagraph*{Algorithm analysis.} The pseudocode for our algorithm is shown in \cref{alg:simontree,alg:findnode,alg:splitnode}.
\cref{oneWord} states the main result of this section.
While the correctness of the algorithm follows mainly from the explanations above,
its linear running time requires an amortized analysis.
We observe that for each position $i$ in $w$ we 
traverse $t$ nodes (representing open blocks) while going up on the leftmost branch,
then insert one leaf on the leftmost branch while closing the $t$ traversed nodes 
and moving them all to the right of the inserted leaf (so out from the leftmost branch).
As the total number of nodes in the Simon-Tree $T_w$ is linear in $|w|$,
and each node is inserted once and traversed once,
the conclusion follows.
For the interested reader we point out
that our analysis resembles to a certain extent the one of the algorithm constructing the Carthesian-Tree for a set of numbers \cite{Vuillemin80}. 

\begin{restatable}{theorem}{restateOneWord}\label{oneWord}
	Given a word $w$, with $|w|=n$,
	we can construct its Simon-Tree in $O(n)$ time.\looseness=-1
\end{restatable}
\begin{proof}	
	The correctness of the construction algorithm follows from its definition which is given in the Algorithm Description paragraph in \cref{sec:algorithm-description}.
	The fundamental conclusions along the way are mainly based on \cref{lem:split}.
	
	To analyse the runtime of the algorithm,
	we observe that for each position $i$ in $w$ we do the following.
	Traverse $t$ nodes (representing open blocks)
	while going up on the leftmost branch,
	then insert one leaf on the leftmost branch
	while closing the $t$ traversed nodes from the leftmost branch
	(they will now be all right from the inserted leaf).
	
	Summing up over all positions in $w$,
	we get that every node of the final Simon-Tree is created once as an open block,
	which can be done in constant time,
	thereafter it is traversed (and closed) once.
	
	This gives us a total linear runtime.
\end{proof}



\section{Connecting two Simon-Trees}
\label{sec:connecting-two-simon-trees}
In this section, we propose  a linear-time algorithm for the \maxkproblem\ problem.
The general idea of this algorithm is to analyse simultaneously the Simon-Trees of the two input words $s$ and $t$ of length $n$ and $n'$, respectively,
and establish a connection between their nodes.

\subsection{The S-Connection}
\label{sec:s-connection}
In our solution of \maxkproblem,
we construct a relation called S-Connection
(abbreviation for Simon-Connection)
between the nodes of the Simon-Trees $T_s$ and~$T_t$
constructed from the two input words $s$ and $t$. 

\begin{definition}
	The (explicit or implicit) $k$-node $a$ of $T_s$
	and the (explicit or implicit) $k$-node $b$ of $T_t$ are S-connected
	(i.e., the pair $(a,b)$ is in the S-Connection)
	if and only if $s[i:n]\sim_k t[j:n']$
	for all positions $i$ in block $a$ and positions $j$ in block $b$.
\end{definition}

If two $k$-nodes $a$ and $b$ are S-connected,
we say that $b$ is $a$'s  S-Connection (and vice versa).
Additionally, if two nodes are S-connected,
then the corresponding blocks are said to be S-connected too.\looseness=-1

Adapted from the equivalence classes within a word,
each explicit or implicit $k$-node of $T_s$ can be S-connected
to at most one $k$-node of $T_t$
(since they are then representing blocks which on their part represent the same equivalence class of the set of suffixes of $w$ w.r.t. $\sim_k$).

\begin{remark}\label{rem:non-cross}
	The $S$-Connection is {\em non-crossing}.
	This means that
	if the $k$-block $a=[m_a:n_a]$ of $T_s$ is S-connected to the $k$-block $b=[m_b:n_b]$ of $T_t$,
	the $k'$-block $c=[m_c:n_c]$ of $T_s$ is S-connected to the $k'$-block $d=[m_d:n_d]$ of $T_t$,
	and $m_a<m_c$,
	then $m_b<m_d$.
	Similarly, if $n_a<n_c$ then $n_b<n_d$.\looseness=-1
\end{remark}

\subsection{The P-Connection}
For constructing the S-Connection efficiently,
we define a coarser relation called P-Connection
(abbreviation for potential-connection)
that covers the S-Connection.
The P-Connection defines, for each node of $T_s$, a unique node of $T_t$
to which it may be S-connected.
Later, we will attempt to determine and {\em split},
for each level $k$ from $1$ to maximally $n$,
all pairs of (explicit and implicit) $k$-nodes
which were P-connected but are not S-connected.
In a sense, this splitting allows us to gradually refine the P-Connection
until we get exactly the S-Connection.
The P-Connection for the words $s$ and $t$ is defined as follows.

\begin{definition}\label{def:pcon}
	The $0$-nodes of $T_s$ and $T_t$ are P-connected.
	For all levels $k$ of $T_s$,
	if the explicit or implicit $k$-nodes $a$ and $b$
	(from $T_s$ and $T_t$, respectively)
	are P-connected,
	then the $i^{th}$ child of $a$ is P-connected to the $i^{th}$ child of $b$, for all $i$.
	No other nodes are P-connected.
\end{definition}

If $k$-nodes $a$ and $b$ are P-connected,
we say that $b$ is $a$'s P-Connection (and vice versa).

According to its definition, the P-Connection can be computed efficiently in a straightforward manner. This definition is essentially based on the following \cref{lem:pcon}.
However, because \cref{lem:pcon} is not both necessary and sufficient
(unlike, e.g., \cref{lem:split}),
it can only be used to define a relation coarser than the S-Connection
and cannot be used to characterise (and, consequently, compute in a simple way) the S-Connection itself.
Recall that in Simon-Trees the children of a node are numbered right to left.\looseness=-1

\begin{restatable}{lemma}{restatePcon}\label{lem:pcon}
	Let $k\geq 1$.
	Let $a = [m_a : n_a]$ be a $k$-block in $s$
	and $b = [m_b : n_b]$ a $k$-block in $t$ with $a \sim_k b$.
	Then the $i^{th}$ child of the node $a$ of $T_s$ can only be S-connected
	(but it is not necessarily connected)
	to the $i^{th}$ child of the node $b$ of $T_t$, for all $i\geq 1$.
\end{restatable}
\begin{proof}
	Let us consider $a'=[m_{a'}:n_{a'}]$,
	the $i^{th}$ child of the node $a$.
	According to \cref{lem:split},
	we have that $i=|\alp(s[n_{a'}:n_a-1])|+1$.
	Let us now assume, for the sake of a contradiction,
	that  $a'$ is $S$-connected to the node $b'=[m_{b'}:n_{b'}]$,
	the $j^{th}$ child of $b$ with $j\neq i$.
	We will show how this leads to a contradiction for $j<i$.
	A similar proof works for $j>i$.
	Clearly, $j=|\alp(t[n_{b'}:n_b-1])|+1$.
	As $j<i$, it follows that there exists a letter $x\in \alp(s[n_{a'}:n_a-1])\setminus \alp(t[n_{b'}:n_b-1])$.
	Let $g=\nextpos (n_{a'},x)$ and $g'=\nextpos (n_{b'},x)$.
	We have that $g\leq n_a-1$ and $g'> n_b-1$.
	This means that $s[g+1:n]\nsim_k t[g'+1:n']$,
	so there exists a subsequence $w$ that is in only one of the words $s[g+1:n]$ and $t[g'+1:n']$.
	As such, $xw$ occurs in only one of the words $s[g:n]$ and $t[g':n']$,
	and, consequently, in only one of the words  $s[n_{a'}:n]$ and $t[n_{b'}:n']$.
	Thus, we have reached a contradiction:
	$a'$ and $b'$ are not $S$-connected.
	Accordingly, we need to have $j=i$.
	
	The same proof actually shows the stronger result:
	For $i>1$,
	if the $i^{th}$ child of the node $a$ is the node $a'=[m_{a'}:n_{a'}]$
	and the $i^{th}$ child of the node $b$ is the node $b'=[m_{b'}:n_{b'}]$
	and $a'\sim_{k+1} b'$,
	then $\alp(s[n_{a'}:n_a-1])=\alp(t[n_{b'}:n_b-1])$.
\end{proof}

It is not hard to see that,
in the spirit of \cref{rem:non-cross},
the P-Connection is non-crossing.
Moreover, by \cref{lem:pcon},
if the $k$-blocks $a$ and $b$ are S-connected,
they are also P-connected.
It is very important to note that a pair of nodes
whose parent-nodes are not S-connected
is also not S-connected.
So, as our approach is to refine the P-Connection till the S-Connection is reached,
we can immediately decide that a pair of nodes $(a,b)$ is not in the S-Connection
when the pair consisting of their respective parent-nodes is not in the S-Connection. 

\subsection{From P- to S-Connection}
\subparagraph*{Preliminary transformation.}
As mentioned, our algorithm solving \maxkproblem\ uses the Simon-Trees of $s$ and $t$.
To make the exposure simpler,
we make the following simple transformation of the trees.
If $a$ is a $k$-node such that $a$ is a singleton,
we add as a child of this node a $(k+1)$-node representing the same block $a$
(this was an implicit node before, now made explicit);
the newly added node on level $k+1$ does not have any children
(i.e., this procedure is not applied recursively).
Before, by \cref{lem:split},
all blocks $[i:i]$ of $w$ appeared explicitly exactly once in $T_w$.
Therefore,
each singleton-block $[i:i]$ of $s$ (respectively, $t$) appears now exactly twice in $T_s$ (respectively, $T_t$).

In general, these now explicit nodes are used to guarantee the existence of a P-connected node (implicit or explicit) for every explicit singleton node on some level $k+1$ that was on a splitting position on level $k$,
so we can determine singleton nodes that are $\sim_k$- but not $\sim_{k+1}$-congruent to the corresponding nodes in the other tree.
The transformation has the following direct consequence that we will use:
each singleton-block $a$ appears now on two consecutive levels.
While the node corresponding to $a$ on the higher level may be S-connected to a node corresponding to a non-singleton-block,
the node corresponding to $a$ on the lower level may be S-connected only to a singleton-node.

As a second consequence,
it is worth noting that explicit nodes might be connected to implicit nodes, too.
However, this is only true for explicit nodes which were added during the transformation described above,
i.e., singleton explicit nodes.
Explicit nodes which are not singletons cannot be connected to implicit nodes.

\subparagraph*{Refining the P-Connection.}
The main step of our approach is,
while considering the levels of the trees $T_s$ and $T_t$ in increasing order,
to identify the pairs of P-connected nodes from the respective levels
which are not S-connected and consequently {\em split} them.
At the same time, we identify the pairs of singleton-blocks occurring explicitly on higher levels
(and only implicitly on the current levels)
which are not S-connected on this level,
and also {\em split} them {\em on the current level}.
For simplicity of exposure, when we split two $k$-blocks,
we say that we $k$-split them.
In order to implement this idea,
we use the following \cref{lem:2words-block-equivalence} to define a splitting criterium.\looseness=-1

We introduce first some notations. 
For $\w\in \{s,t\}$, a position $j \leq |\w|$, and a letter $x$,
we define $\nextpos_\w(j,x)$ as the leftmost position
where $x$ occurs in $\w[j:|\w|]$,
or as $\infty$ if $x\notin \alp(\w[j:|\w|])$.
For a block $a=[m_a:n_a]$ of the word $\w$ and a letter $x$,
we define $\nextpos_\w(a,x)=\nextpos_\w(n_a,x)$.
We generally omit the subscript $\w$ when it is clear from the context.
Furthermore, we define $\alp(a)$ for a block $a = [m_a:n_a]$ as $\alp(w[m_a:n_a])$.

\begin{restatable}{lemma}{restateBlockEq}\label{lem:2words-block-equivalence}
	Let $k\geq 1$.
	Let $a = [m_a : n_a]$ be a $k$-block in the word $s$
	and $b = [m_b : n_b]$ a $k$-block in the word $t$
	with $a \sim_k b$.
	Let $a' = [m_{a'} : n_{a'}]$ be a $(k+1)$-block in $a$ 
	and $b' = [m_{b'} : n_{b'}]$ be a $(k+1)$-block in $b$.
	Then $a' \nsim_{k+1} b'$ if and only if there exists a letter $x$ such that 
	$s[\nextpos(a', x) + 1 : n] \nsim_k t[\nextpos(b', x) +1 : n']$.
\end{restatable}
\begin{proof}
	For completeness, we first introduce a notation. 
	For $w\in \{s,t\}$, a position $j\in w$, and a letter $x$,
	we define $\nextpos_w(j,x)$ as the leftmost position
	where $x$ occurs in $w[j:|w|]$,
	or as $\infty$ if $x\notin \alp(w[j:|w|])$.
	For a block $a=[m_a:n_a]$ of the word $w$ and a letter $x$,
	we define $\nextpos_w(a,x)=\nextpos(n_a,x)$.
	We generally omit the subscript $w$ as it is clear from the context to which we refer.
	
	Let us assume $a' \nsim_{k+1} b'$
	and show that there exists a letter $x$
	such that $s[\nextpos(a', x) +1 : n] \nsim_k t[\nextpos(b', x) + 1 : n']$.
	As $a' \nsim_{k+1} b'$,
	we can assume without loss of generality that there exists a word $w$, with $|w|=k+1$,
	such that $w \in \SF_{k+1}(n_{a'}, s) \setminus \SF_{k+1}(n_{b'}, t)$.
	Let $x$ be the first letter of $w$, i.e., $w=xw'$ for some word $w'$. 
	If $s[\nextpos(a', x) + 1:n] \sim_k t[\nextpos(b', x) + 1:n']$ would hold,
	then $w'$ would be a subsequence of length $k$,
	in both $s[\nextpos(a', x) + 1:n]$ and $t[\nextpos(b', x) + 1:n']$.
	Thus, $w=xw'$ would be a subsequence of both $s[\nextpos(a', x):n]$ and $t[\nextpos(b', x):n']$, so $w\in \SF_{k+1}(n_{b'}, t)$.
	Thus, $s[\nextpos(a', x)+1:n] \nsim_k t[\nextpos(b', x)+1:n']$
	
	Now, let there be a letter $x$
	such that $\SF_k(\nextpos(n_{a'}, x)+1, s) \neq \SF_k(\nextpos(n_{b'}, x)+1, t)$.
	Then we can assume without losing generality that there exists $w'$ of length $k$ in $\SF_k(\nextpos(n_{a'}, x)+1, s) \setminus \SF_k(\nextpos(n_{b'}, x)+1, t)$.
	Clearly, $w=xw'$ is in $\SF_k(\nextpos(n_{a'}, x), s)$ $\subseteq \SF_k(n_{a'}, s)$
	but $w=xw' \notin \SF_k(\nextpos(n_{b'}, x), t)$,
	so $w \notin \SF_k(n_{b'}, t)$.
	This means that $s[n_{a'}:n] \nsim_{k+1} t[n_{b'}:n']$.
\end{proof}

\begin{figure}
	\centering
	\tikzset{ns/.style={
		inner sep = 0mm, outer sep=0mm,
		append after command={
			[very thick]
			(\tikzlastnode.north west)edge(\tikzlastnode.north east)
			[very thick]
			(\tikzlastnode.south west)edge(\tikzlastnode.south east)
		}
	}
}

\tikzset{wens/.style={
		inner sep = 0mm, outer sep=0mm,
		append after command={
			[very thick]
			(\tikzlastnode.north west)edge(\tikzlastnode.south west)
			[very thick]
			(\tikzlastnode.north east)edge(\tikzlastnode.south east)
			[very thick]
			(\tikzlastnode.north west)edge(\tikzlastnode.north east)
			[very thick]
			(\tikzlastnode.south west)edge(\tikzlastnode.south east)
		}
	}
}

\tikzset{wns/.style={
		inner sep = 0mm, outer sep=0mm,
		append after command={
			[very thick]
			(\tikzlastnode.north west)edge(\tikzlastnode.south west)
			[very thick]
			(\tikzlastnode.north west)edge(\tikzlastnode.north east)
			[very thick]
			(\tikzlastnode.south west)edge(\tikzlastnode.south east)
		}
	}
}

\tikzset{ens/.style={
		inner sep = 0mm, outer sep=0mm,
		append after command={
			[very thick]
			(\tikzlastnode.north east)edge(\tikzlastnode.south east)
			[very thick]
			(\tikzlastnode.north west)edge(\tikzlastnode.north east)
			[very thick]
			(\tikzlastnode.south west)edge(\tikzlastnode.south east)
		}
	}
}
\begin{tikzpicture}
	\Large
	\node[circle] (s0) at (0,0) {$s$};
	\node[ns, minimum height=6mm, minimum width=8mm, align=center, right= 0mm of s0] (s1) {$\ldots$};
	\node[wns, minimum height=6mm, minimum width=12mm, align=center, right= 0mm of s1] (s2) {$a^\prime$};
	\node[ens, minimum height=6mm, minimum width=3mm, align=center, right= 0mm of s2] (singa) {$ $};
	\node[ns, minimum height=6mm, minimum width=8mm, align=center, right= 0mm of singa] (s3) {$\ldots$};
	\node[ns, minimum height=6mm, minimum width=3mm, align=center, right= 0mm of s3] (s4) {$x$};
	\node[wens, minimum height=6mm, minimum width=3mm, align=center, right= 0mm of s4] (s5) {$ $};
	\node[circle, minimum height=6mm, minimum width=3mm, align=center, above= 0mm of s5] (s5aboce) {$i$};
	\node[ns, minimum height=6mm, minimum width=8mm, align=center, right= 0mm of s5] (s6) {$\ldots$};

	\node[circle] (t0) at (7,0) {$t$};
	\node[ns, minimum height=6mm, minimum width=8mm, align=center, right= 0mm of t0] (t1) {$\ldots$};
	\node[wns, minimum height=6mm, minimum width=12mm, align=center, right= 0mm of t1] (t2) {$b^\prime$};
	\node[ens, minimum height=6mm, minimum width=3mm, align=center, right= 0mm of t2] (singb) {$ $};
	\node[ns, minimum height=6mm, minimum width=8mm, align=center, right= 0mm of singb] (t3) {$\ldots$};
	\node[ns, minimum height=6mm, minimum width=3mm, align=center, right= 0mm of t3] (t4) {$x$};
	\node[wens, minimum height=6mm, minimum width=3mm, align=center, right= 0mm of t4] (t5) {$ $};
	\node[circle, minimum height=6mm, minimum width=3mm, align=center, above= 0mm of t5] (t5above) {$j$};
	\node[ns, minimum height=6mm, minimum width=8mm, align=center, right= 0mm of t5] (t6) {$\ldots$};

	\node[wns, minimum height=6mm, minimum width=12mm, align=center, above= 15mm of s2] (s2u) {$a^\prime$};
	\node[ens, minimum height=6mm, minimum width=3mm, align=center, right= 0mm of s2u] (singua) {$ $};
	
	\node[wns, minimum height=6mm, minimum width=12mm, align=center, above= 15mm of t2] (t2u) {$b^\prime$};
	\node[ens, minimum height=6mm, minimum width=3mm, align=center, right= 0mm of t2u] (singub) {$ $};
	
	\node[wens, minimum height=6mm, minimum width=20mm, align=center, above= 15mm of s2u] (a) {$a$};
	\node[minimum height=6mm, minimum width=8mm, align=center, left= 0mm of a] (aldots) {$\ldots$};
	\node[minimum height=6mm, minimum width=8mm, align=center, right= 0mm of a] (ardots) {$\ldots$};
	\node[minimum height=6mm, minimum width=8mm, align=center, right= 7mm of ardots] (klevel) {$level-k$};
	
	\node[wens, minimum height=6mm, minimum width=20mm, align=center, above= 15mm of t2u] (b) {$b$};
	\node[minimum height=6mm, minimum width=8mm, align=center, left= 0mm of b] (bldots) {$\ldots$};
	\node[minimum height=6mm, minimum width=8mm, align=center, right= 0mm of b] (brdots) {$\ldots$};

	\draw[bend right=70, decoration={markings, mark=at position .5  with {\draw [black, thick,-] 
					++ (-0.3,-0.3) 
					-- (0.3,0.3);}}, postaction=decorate] (s5.south) to[out=-30,in=-90] node[above]{$\quad\nsim_k$} (t5.south);
	
	
	\draw[bend right, decoration={markings, mark=at position .5  with {\draw [black, thick,-] 
			++ (-0.3,-0.3) 
			-- (0.3,0.3);}}, postaction=decorate] (singua.south east)  to[out=-30,in=-150] node[above,yshift=3mm,xshift=8mm]{$\nsim_{k+1}$} (t2u.south west);
	\draw[bend right] (a.south east) edge node[above]{$\sim_k$} (b.south west);
	\draw[-] (a.south) edge  (s2u.north);
	\draw[-] (b.south) edge  (t2u.north);
	
	\draw[dashed] (s2u.south west) edge  (s2.north west);
	\draw[dashed] (singua.south east) edge  (singa.north east); 
	
	\draw[dashed] (t2u.south west) edge  (t2.north west);
	\draw[dashed] (singub.south east) edge  (singb.north east); 
	
	\draw[dotted, thick] (singa.south west) |- ++ (0,-0.2) -| (s4.south east); 
	\draw[dotted, thick] (singb.south west) |- ++ (0,-0.2) -| (t4.south east); 
\end{tikzpicture}
	\caption{Illustration of \cref{lem:2words-block-equivalence}.}
	\label{fig:lem-2words-block-equivalence}
\end{figure}

The main idea of this lemma
(illustrated in \cref{fig:lem-2words-block-equivalence})
is that two $(k+1)$-blocks $a'$ and $b'$ are not S-connected,
although their parents were S-connected,
if and only if we can find a letter $x$
such that $s[\nextpos(a', x) + 1 : n] $ and $ t[\nextpos(b', x) + 1 : n']$
are not $k$-equivalent but $(k-1)$-equivalent.
That is, $\nextpos(a', x) + 1 $ and $\nextpos(b', x) + 1$ should occur,
respectively, in two $k$-blocks which were split,
but whose parents were S-connected.
A word distinguishing the suffixes starting in $a'$ from those starting in $b'$ has the first letter $x$,
and is continued by the word of length $k$ which distinguishes $s[\nextpos(a', x) + 1 : n]$ and $t[\nextpos(b', x) +1 : n']$. 

\subparagraph*{Identifying P-connected pairs to be split.}
\label{sec:identify-pairs}
When going through the trees level by level,
the $1$-blocks
(all occuring explicitly on level $1$ of $T_s$ and respectively $T_t$)
which are S-connected can be easily and efficiently identified:
the $i^{th}$ node $a=[m_a:n_a]$ on level $1$ of $T_s$ is connected to the $i^{th}$ node $b=[m_b:n_b]$ of $T_t$
if and only if $\alp(s[n_a:n])=\alp(t[n_b:n'])$.
All the other P-connected pairs of $1$-blocks are not S-connected,
so they are $1$-split.

The identification of the pairs of $(k+1)$-blocks and pairs of singletons which need to be $(k+1)$-split is based on \cref{lem:2words-block-equivalence}.
The idea is the following. 
A pair of P-connected $(k+1)$-blocks $a' = [m_{a'} : n_{a'}]$ of $T_s$
and $b' = [m_{b'} : n_{b'}]$ of $T_t$ 
is not S-connected
if and only if there exists a letter $x$
such that $s[\nextpos(a', x) + 1 : n] \nsim_k t[\nextpos(b', x) +1 : n']$.
So, in order to be able to $(k+1)$-split two nodes (whose parents are S-connected),
we need to identify two positions $i$ and $j$ (and a corresponding letter $x$),
with $i=\nextpos(a', x) + 1$ and $j=\nextpos(b', x)+1$
which were $k$-split but not $(k-1)$-split.
We search for position $i$ inside the $k$-blocks of $T_s$,
and try to see where position $j$ may occur in the blocks of $T_t$
such that these two positions are not in S-connected $k$-blocks.
To find the position $i$ (and the corresponding $j$) we analyse two cases.\looseness=-1

\begin{description}
	\item[The first case (A)]
	is when $i$ occurs inside an implicit $k$-node,
	which is the singleton-$k$-block $[i:i]$.
	On the lowest level where this block appeared as an explicit node,
	it was S-connected to a node $[j:j]$ representing a singleton too,
	according to \cref{sec:s-connection} and \cref{lem:split}.
	Thus, position $i$ can only be $k$-split from the position $j$ of $t$ to which it was S-connected
	(it was already disconnected from all other positions on level $\ell$).
	If $i$ and $j$ are both directly preceded by the same symbol (say $x$),
	then the pair $(i,j)$ gives us exactly the positions we were searching for.
	\item[The second case (B)]
	is when $i$ occurs inside an explicit $k$-node in $T_s$.
	Let $A=[m_A:n_A]$ and $B=[m_B:n_B]$ be two $(k-1)$ blocks from $T_s$ and $T_t$, respectively,
	such that $A\sim_{k-1} B$, and $a=[m_a:n_a]$ and $b=[m_b:n_b]$ be the $\ell^{th}$ child of $A$ and $B$, respectively.
	Clearly, $b$ might be explicit, implicit, or even empty.
	If $b$ is non-empty, the following holds.
	All positions of $a$ are $k$-split from the positions $[m_B:m_b-1]$ and from the positions $[n_b+1:n_B]$,
	because $a$ is not P-connected to the blocks covering those positions.
	Also, if $a$ and $b$ are not S-connected
	then all positions of $a$ are also $k$-split from the positions of $b$.
\end{description}%
\begin{figure}[!htb]
	\centering
	\includegraphics[width=\linewidth]{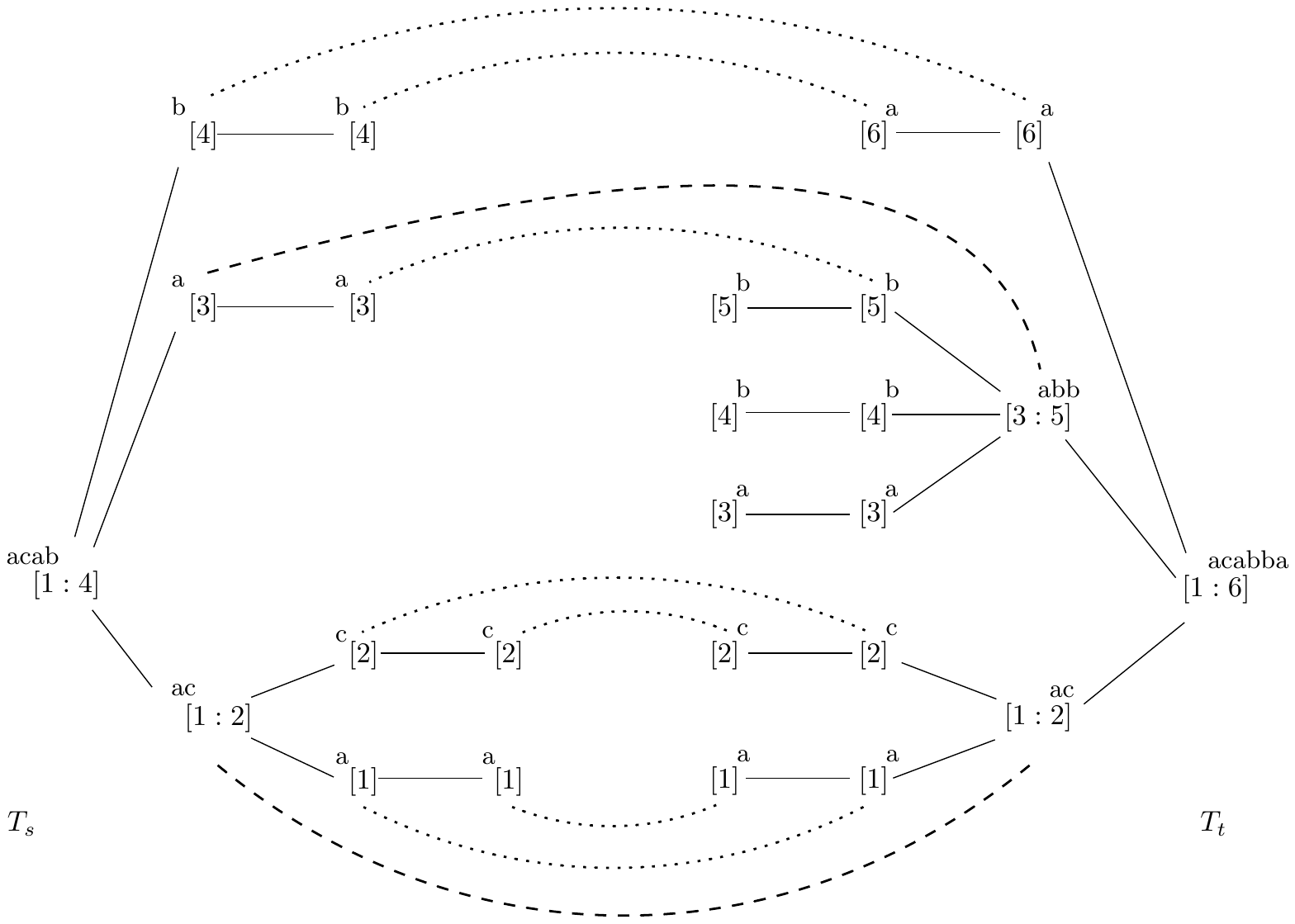}
	\caption{Two transformed Simon-Trees with their P- and S-Connection.}
	\label{fig:two-simon-trees}
\end{figure}
\subparagraph*{An example.} Let $s=acab$ and $t=acabba$ be two words.
Their respective Simon-Trees $T_s$ and $T_t$ are depicted in \cref{fig:two-simon-trees} along with their P- and S-Connection.
Note that for the sake of not crossing edges and simplifying the presentation,
the Simon-Trees in \cref{fig:two-simon-trees} are rotated by 90 degrees to the right and to the left, respectively.
Thus, the roots of the trees are on the outer left and right side of the figure.
Additionally, the tree $T_t$ on the right is mirrored,
so that nodes from a P-connected pair are vis-à-vis.
Also, the trees already contain the singleton nodes that were originally implicit but are now made explicit by our aforementioned transformation.
From the figures it becomes also clear
that this transformation is needed in the case of the singleton-node $[5:5]$ from the $2$nd level of $T_t$
which is P-connected to the singleton-node $[3:3]$ from the $2$nd level of $T_s$.

In the beginning we are considering all possibly connected blocks by determining all P-connected pairs.
While the dashed and dotted lines connect, respectively, the nodes of all the pairs from the P-Connection,
the S-Connection is obtained by splitting step by step P-connected pairs that cannot be equivalent with regard to their respective level.
The dotted edges symbolize exactly these split pairs,
and in the end, the S-Connection consists only of the pairs connected with a dashed edge.

Following \cref{thm:main,thm:main1} stated at the very end of this paper,
we get the largest $k$ for which the two words $s$ and $t$ are $k$-equivalent by finding the largest $k$ for which the $k$-blocks containing position $1$ of both words are S-connected.
In our example, the blocks $s[1:4]$ and $t[1:6]$ representing the complete words are naturally $0$-equivalent.
Furthermore, as seen in \cref{fig:two-simon-trees}, the blocks $s[1:2]$ and $t[1:2]$ on level $1$ are S-connected,
but the blocks $s[1:1]$ and $t[1:1]$ are not S-connected on level $2$.
Thus, the largest $k$ for which $s \sim_k t$ holds is $1$. \looseness=-1

\subparagraph*{Path to efficiency.}
Taking $(i,j)$ to be each position of block $a$ paired with each of the positions from which it is $k$-split,
according to the above, might not be efficient.
However, the combinatorial \cref{lem:2words-block-equivalence} allows us to switch slightly the point of view and ultimately obtain an efficient method. 
We will traverse the $k$th level of a Simon-Tree $T_s$ from right to left and when considering a node $a$ on this level,
our approach is to determine which nodes should be $k+1$-split due to any pair $(i,j)$, where $i$ is some position occurring in the block $a$.
The mechanism allowing us to do this is stated in \Cref{lem:aksplit}, and this essentially explains how to determine all the $k+1$-splits determined by positions of $a$ (and their corresponding pairing from $b$). We show how to proceed in both cases (A) and (B). Clearly, a symmetrical approach would also work (so looking at nodes in $T_t$ and positions in $t$).\looseness=-1

Firstly, we need a few more notations.
For a block $a=[m_a:n_a]$ of $s$ or $t$ and a letter $x$,
let $\prevpos(a,x)$ be the rightmost occurrence of $x$ in $s[1:m_a-2]$
(or $0$ if $x\notin \alp(s[1:m_a-2])$), 
and let $\rightpos (a,x)$ be the rightmost occurrence of $x$ in $s[1:n_a-1]$
(or $0$ if $x\notin \alp(s[1:n_a-1])$). 

The setting in which \cref{lem:aksplit} is stated is the following.
We have two P-connected $k$-blocks $a=[m_a:n_a]$ and $b=[m_b:n_b]$ from $T_s$ and $T_t$, respectively,
whose parent-nodes (explicit or implicit) are S-connected. 
The lemma defines a necessary and sufficient condition for a pair of (explicit or implicit) $(k+1)$-nodes $(a',b')$
to be $(k+1)$-split because there exists a letter $x$ and a pair of positions $(i,j)$,
with $i=\nextpos(a', x) + 1$, $i\in [m_a:n_a]$, $j=\nextpos(b', x)+1$,
and $s[i:n]\nsim_k t[j:n']$.
Such a pair $(a',b')$ is called $(a,k+1)$-split
(that is, $a$ causes the respective split on level $k+1$).
Note that $a'$ and $b'$ are the (explicit or implicit) children of
either $a$ and $b$
or of two $k$-blocks which are left of $a$ and $b$.
In any way, their parents are S-connected;
otherwise $a'$ and $b'$ would have already been split on a higher level.

This setting is also illustrated in \cref{fig:lem-aksplit}.\looseness=-1

\begin{restatable}{lemma}{restateAksplit}\label{lem:aksplit}
	
	For $k\geq 1$, 
	let $a=[m_a:n_a]$ be a $k$-block in $s$
	and $b=[m_b:n_b]$ its P-Connection
	(which is a $k$-block in $t$).
	Then a pair of P-connected $(k+1)$-blocks
	$a' = [m_{a'} : n_{a'}]$ in $s$
	and $b' = [m_{b'} : n_{b'}]$ in $t$
	is $(a,k+1)$-split
	if and only if there exists a letter $x$ in $\alp(s[m_a-1:n_a-1])$ such that at least one of the following holds:
	\begin{enumerate}
		\item $a'$  ends strictly between $\prevpos(a,x)$ and $m_a$
		(i.e., $\prevpos(a,x) < n_{a'}< m_a$),
		and $b'$ ends to the left of $\prevpos(b,x)+1$
		(i.e., $n_{b'} \leq \prevpos(b,x)$). 
		\item $a'$  ends between $\prevpos(a,x)$ and $\rightpos(a,x)$
		(i.e., $\prevpos(a,x) < n_{a'}$ $\leq$ $\rightpos(a,x)$),
		and $b'$ ends between $\rightpos(b,x)$ and $n_b$
		(i.e., $\rightpos(b,x)< n_{b'} \leq n_b$).
		\item $a\nsim_k b$, $a'$  ends between $\prevpos(a,x)$ and $\rightpos(a,x)$
		(i.e., $\prevpos(a,x) < n_{a'}\leq \rightpos(a,x)$),
		and $b'$ ends between $\prevpos(b,x)$ and $\rightpos(b,x)$
		(i.e., $\prevpos(b,x)< n_{b'} \leq \rightpos(b,x)$).
	\end{enumerate}
\end{restatable}
\begin{proof}
	Assume $A=[m_A:n_A]$ and $B=[m_B,n_B]$ are two $(k-1)$ (implicit or explicit) blocks from $T_s$ and $T_t$, respectively,
	and $a=[m_a:n_a]$ and $b=[m_b,n_b]$ the $\ell^{th}$ children of $A$ and $B$
	(if $A$ or, respectively, $B$ is implicit, then $a$ or, respectively, $b$ is also implicit). 
	
	Let us consider a pair of P-connected $(k+1)$-nodes
	$a' = [m_{a'} : n_{a'}]$ in $s$ and
	$b' = [m_{b'} : n_{b'}]$ in $t$ which is $(a,k+1)$-split.
	Then there exist two positions $i$ and $j$,
	with $i=\nextpos(a', x) + 1$ and $j=\nextpos(b', x)+1$
	which were $k$-split
	but whose $(k-1)$-blocks are S-connected,
	and, moreover, $i$ is in $a$.
	This means $j$ must have been in $B$.
	There are three cases to analyse.
	
	\begin{case}{$j=\nextpos(b',x)+1$ is a position of $[m_B:m_b-1]$.}
		As $i=\nextpos(a',x)+1$ is in $a$,
		we get that $a'$ ends after $\prevpos(a,x)$.
		Also, $b'$ is P-connected to $a'$ and the P-Connection is non-crossing,
		so $a'$ should also end before $m_a$.
		Moreover, $b'$ ends before $j=\nextpos(b',x)+1$,
		and, as $\nextpos(b',x)$ returns the position of a letter $x$ in the interval $[m_B-1:m_b-2]$,
		we get that $j \leq \prevpos(b,x)+1$. 
	\end{case}
	
	\begin{case}{$j=\nextpos(b',x)+1$ is a position of $[n_b+1:n_B]$.}
		The analysis of where $a'$ may occur is similar to the first case above:
		as $i=\nextpos(a',x)+1$ is in $a$,
		we get that $a'$ ends after $\prevpos(a,x)$.
		As $i$ is in $a$,
		we get that $a'$ cannot end to the right of the rightmost $x$ in $a$,
		so $\prevpos(a,x) < n_{a'}\leq \rightpos(a,x)$.
		Now, $\nextpos(b',x)+1$ is in $[n_b+1:n_B]$,
		which means that $\nextpos(b',x)$ is not in $b$,
		so $\nextpos(b',x)>\rightpos(b,x)$.
		As $b'$ is P-connected to $a$,
		and the P-Connection is non-crossing,
		we get that $b'$ cannot end to the right of $n_b$.
	\end{case} 
	
	\begin{case}{$j$ is inside $b$ and $a\nsim_k b$
			(so that $i$ and $j$ are $k$-split).}
		The analysis is, however, very similar to the above,
		and the conclusion follows in the same way.
	\end{case}
	
	We will now show the converse. 
	In case 1, we have that $\nextpos (a',x)+1$ is $k$-split from $\nextpos(b',x)+1$,
	as the former is a position inside $a$
	and the latter is a position to the left of $m_b-1$.
	In case 2, we have that $\nextpos (a',x)+1$ is $k$-split from $\nextpos(b',x)+1$,
	as the first is a position inside $a$
	and the second one is a position strictly to the right of $n_b$.
	In case 3, we have that $\nextpos (a',x)+1$ is $k$-split from $\nextpos(b',x)+1$,
	as $\nextpos (a',x)+1$ is a position inside $a$
	and $\nextpos(b',x)+1$ is a position inside $b$,
	and $a$ and $b$ are $k$-split. 
\end{proof}

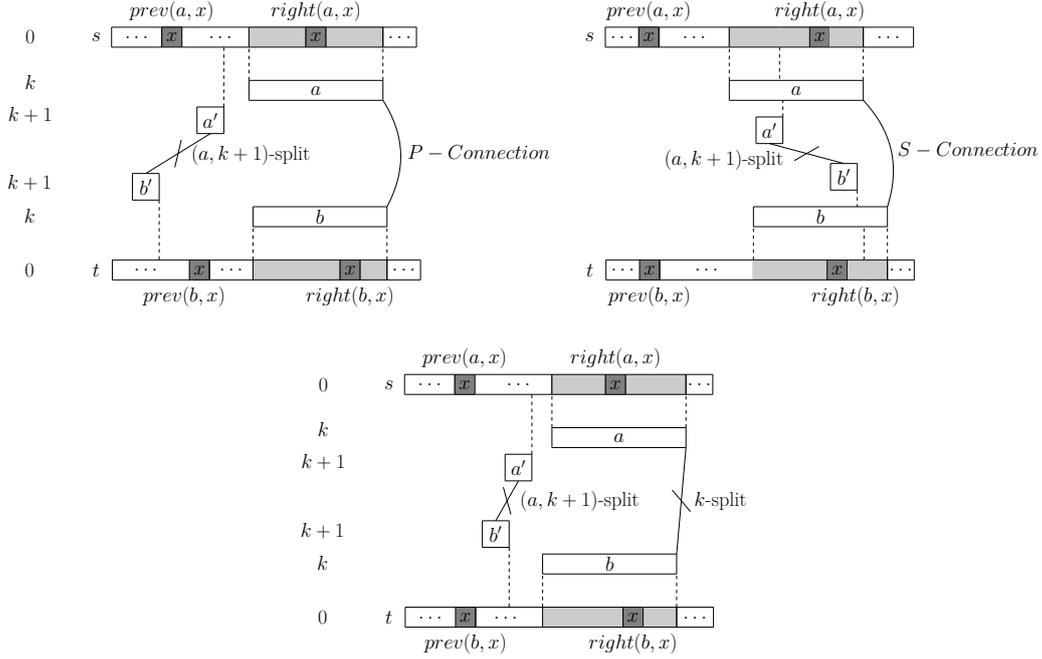
\begin{figure}
	\centering
	\begin{subfigure}{0.54\linewidth}
		\centering
		\tikzset{ns/.style={
		inner sep = 0mm, outer sep=0mm,
		append after command={
			[very thick]
			(\tikzlastnode.north west)edge(\tikzlastnode.north east)
			[very thick]
			(\tikzlastnode.south west)edge(\tikzlastnode.south east)
		}
	}
}

\tikzset{wens/.style={
		inner sep = 0mm, outer sep=0mm,
		append after command={
			[very thick]
			(\tikzlastnode.north west)edge(\tikzlastnode.south west)
			[very thick]
			(\tikzlastnode.north east)edge(\tikzlastnode.south east)
			[very thick]
			(\tikzlastnode.north west)edge(\tikzlastnode.north east)
			[very thick]
			(\tikzlastnode.south west)edge(\tikzlastnode.south east)
		}
	}
}

\tikzset{wns/.style={
		inner sep = 0mm, outer sep=0mm,
		append after command={
			[very thick]
			(\tikzlastnode.north west)edge(\tikzlastnode.south west)
			[very thick]
			(\tikzlastnode.north west)edge(\tikzlastnode.north east)
			[very thick]
			(\tikzlastnode.south west)edge(\tikzlastnode.south east)
		}
	}
}

\tikzset{ens/.style={
		inner sep = 0mm, outer sep=0mm,
		append after command={
			[very thick]
			(\tikzlastnode.north east)edge(\tikzlastnode.south east)
			[very thick]
			(\tikzlastnode.north west)edge(\tikzlastnode.north east)
			[very thick]
			(\tikzlastnode.south west)edge(\tikzlastnode.south east)
		}
	}
}

\begin{tikzpicture}
	\LARGE
	\node[circle] (t0) at (0,0) {$t$};
	\node[wns, minimum height=6mm, minimum width=20mm, align=center, right= 0mm of t0] (t1) {$\ldots$};
	\node[ns, minimum height=6mm, minimum width=3mm, align=center, right= 0mm of t1] (t2) {$ $};
	\node[wens, fill=black!50, minimum height=6mm, minimum width=6mm, align=center, right= 0mm of t2] (t3) {$x$};
	\node[wens, minimum height=6mm, minimum width=13mm, align=center, right= 0mm of t3] (t4) {$\ldots$};
	\node[wens, fill=black!20, minimum height=6mm, minimum width=26mm, align=center, right= 0mm of t4] (t5) {$ $};
	\node[wens, fill=black!50, minimum height=6mm, minimum width=6mm, align=center, right= 0mm of t5] (t6) {$x$};
	\node[wens, fill=black!20, minimum height=6mm, minimum width=8mm, align=center, right= 0mm of t6] (t7) {$ $};
	\node[wens, minimum height=6mm, minimum width=10mm, align=center, right= 0mm of t7] (t8) {$\ldots$};
	\node[circle, align=center, left= 10mm of t0] (t9) {$0$};
	\node[circle, align=center, above= 6mm of t9] (t10) {$k$};
	\node[circle, align=center, above= -5mm of t10] (t11) {$k+1$};
	
	\node[circle] (s0) at (0,7) {$s$};
	\node[wens, minimum height=6mm, minimum width=15mm, align=center, right= 0mm of s0] (s1) {$\ldots$};
	\node[wens, fill=black!50, minimum height=6mm, minimum width=6mm, align=center, right= 0mm of s1] (s2) {$x$};
	\node[ns, minimum height=6mm, minimum width=3mm, align=center, right= 0mm of s2] (s3) {$ $};
	\node[ens, minimum height=6mm, minimum width=17mm, align=center, right= 0mm of s3] (s4) {$\ldots$};
	\node[wens, fill=black!20, minimum height=6mm, minimum width=17mm, align=center, right= 0mm of s4] (s5) {$ $};
	\node[wens, fill=black!50, minimum height=6mm, minimum width=6mm, align=center, right= 0mm of s5] (s6) {$x$};
	\node[wens, fill=black!20, minimum height=6mm, minimum width=17mm, align=center, right= 0mm of s6] (s7) {$ $};
	\node[wens, minimum height=6mm, minimum width=10mm, align=center, right= 0mm of s7] (s8) {$\ldots$};
	\node[circle, align=center, left= 10mm of s0] (s9) {$0$};
	\node[circle, align=center, below= 3mm of s9] (s10) {$k$};
	\node[circle, align=center, below= -5mm of s10] (s11) {$k+1$};
	
	\node[above=0mm of s2] (prevs1) {$prev(a,x)$}; 
	\node[below=0mm of t2] (prevs2) {$prev(b,x)$}; 
	\node[above=0mm of s6] (right1) {$right(a,x)$};
	\node[below=0mm of t6] (right1) {$right(b,x)$};
	
	\node[wens, minimum height=6mm, minimum width=40mm, xshift=11.5mm, align=center, below= 10mm of s5] (s5aux) {$a$};
	
	\draw[dashed] (s5aux.north west) edge  (s5.south west);
	\draw[dashed] (s5aux.north east) edge  (s7.south east);
	
	\node[wens, minimum height=6mm, minimum width=40mm, xshift=7mm, align=center, above= 10mm of t5] (b) {$b$};
	
	\draw[dashed] (t5.north west) edge  (b.south west);
	\draw[dashed] (t7.north east) edge  (b.south east);
	
	\draw[bend left,thick] (s5aux.south east) edge node[right]{$P-Connection$} (b.north east);
	
	\node[wens, minimum height=8mm, minimum width=8mm, align=center, below= 18mm of s4, xshift=-3mm] (amid1) {$a^\prime$};
	\node[wens, minimum height=8mm, minimum width=8mm, align=center, above= 18mm of t1,] (bmid1) {$b^\prime$};
	
	\draw[dashed] (amid1.north east) -- ++ (0, 1.8);
	\draw[dashed] (bmid1.south east) -- ++ (0, -1.8);
	
	\draw[-, decoration={markings, mark=at position .5  with {\draw [black, thick,-] 
			++ (-0.3,-0.3) 
			-- (0.3,0.3);}}, postaction=decorate] (bmid1.north) -- node[right, xshift=2mm, yshift=-1]{$(a,k+1)$-split} (amid1.south);
\end{tikzpicture}
	\end{subfigure}%
	\begin{subfigure}{0.46\linewidth}
		\centering
		\tikzset{ns/.style={
		inner sep = 0mm, outer sep=0mm,
		append after command={
			[very thick]
			(\tikzlastnode.north west)edge(\tikzlastnode.north east)
			[very thick]
			(\tikzlastnode.south west)edge(\tikzlastnode.south east)
		}
	}
}

\tikzset{wens/.style={
		inner sep = 0mm, outer sep=0mm,
		append after command={
			[very thick]
			(\tikzlastnode.north west)edge(\tikzlastnode.south west)
			[very thick]
			(\tikzlastnode.north east)edge(\tikzlastnode.south east)
			[very thick]
			(\tikzlastnode.north west)edge(\tikzlastnode.north east)
			[very thick]
			(\tikzlastnode.south west)edge(\tikzlastnode.south east)
		}
	}
}

\tikzset{wns/.style={
		inner sep = 0mm, outer sep=0mm,
		append after command={
			[very thick]
			(\tikzlastnode.north west)edge(\tikzlastnode.south west)
			[very thick]
			(\tikzlastnode.north west)edge(\tikzlastnode.north east)
			[very thick]
			(\tikzlastnode.south west)edge(\tikzlastnode.south east)
		}
	}
}

\tikzset{ens/.style={
		inner sep = 0mm, outer sep=0mm,
		append after command={
			[very thick]
			(\tikzlastnode.north east)edge(\tikzlastnode.south east)
			[very thick]
			(\tikzlastnode.north west)edge(\tikzlastnode.north east)
			[very thick]
			(\tikzlastnode.south west)edge(\tikzlastnode.south east)
		}
	}
}

\begin{tikzpicture}
	\LARGE
	\node[circle] (t0) at (0,0) {$t$};
	\node[wns, minimum height=6mm, minimum width=10mm, align=center, right= 0mm of t0] (t1) {$\ldots$};
	\node[wens, fill=black!50,minimum height=6mm, minimum width=6mm, align=center, right= 0mm of t1] (t2) {$x$};
	\node[wens,  minimum height=6mm, minimum width=28mm, align=center, right= 0mm of t2] (t3) {$\ldots$};
	\node[ns, fill=black!20, minimum height=6mm, minimum width=22mm, align=center, right= 0mm of t3] (t4) {$ $};
	\node[wens, fill=black!50, minimum height=6mm, minimum width=6mm, align=center, right= 0mm of t4] (t5) {$x$};
	\node[ns, fill=black!20, minimum height=6mm, minimum width=5mm, align=center, right= 0mm of t5] (t6) {$ $};
	\node[ens, fill=black!20, minimum height=6mm, minimum width=7mm, align=center, right= 0mm of t6] (t7) {$ $};
	\node[wens, minimum height=6mm, minimum width=8mm, align=center, right= 0mm of t7] (t8) {$\ldots$};
	
	\node[circle] (s0) at (0,7) {$s$};
	\node[wens, minimum height=6mm, minimum width=10mm, align=center, right= 0mm of s0] (s1) {$\ldots$};
	\node[wens, fill=black!50, minimum height=6mm, minimum width=6mm, align=center, right= 0mm of s1] (s2) {$x$};
	\node[wens, minimum height=6mm, minimum width=21mm, align=center, right= 0mm of s2] (s3) {$\ldots$};
	\node[wns, fill=black!20, minimum height=6mm, minimum width=15mm, align=center, right= 0mm of s3] (s4) {$ $};
	\node[ns, fill=black!20, minimum height=6mm, minimum width=9mm, align=center, right= 0mm of s4] (s5) {$ $};
	\node[wens, fill=black!50, minimum height=6mm, minimum width=6mm, align=center, right= 0mm of s5] (s6) {$x$};
	\node[ens, fill=black!20, minimum height=6mm, minimum width=10mm, align=center, right= 0mm of s6] (s7) {$ $};
	\node[wens, minimum height=6mm, minimum width=15mm, align=center, right= 0mm of s7] (s8) {$\ldots$};
	
	\node[above=0mm of s2] (prevs1) {$prev(a,x)$}; 
	\node[below=0mm of t2] (prevs2) {$prev(b,x)$}; 
	\node[above=0mm of s6] (right1) {$right(a,x)$};
	\node[below=0mm of t6] (right1) {$right(b,x)$};
	
	\node[wens, minimum height=6mm, minimum width=40mm, xshift=0.5mm, align=center, below= 10mm of s5] (a) {$a$};
	
	\draw[dashed] (a.north west) edge  (s4.south west);
	\draw[dashed] (a.north east) edge  (s7.south east);
	\draw[dashed] (s5.south west) -- ++ (0, -1);
	
	\node[wens, minimum height=6mm, minimum width=40mm, xshift=-5mm, align=center, above= 10mm of t5] (b) {$b$};
	
	\draw[dashed] (t4.north west) edge  (b.south west);
	\draw[dashed] (t7.north east) edge  (b.south east);
	\draw[dashed] (t7.north west) -- ++ (0, 1);
	
	\draw[bend left,thick] (a.south east) edge node[right]{$S-Connection$} (b.north east);
	
	\node[wens, minimum height=8mm, minimum width=8mm, align=center, xshift=-8mm, below= 5mm of a,] (amid1) {$a^\prime$};
	\node[wens, minimum height=8mm, minimum width=8mm, align=center, xshift=7mm, above= 5mm of b,] (bmid1) {$b^\prime$};
	
	\draw[dashed] (amid1.north east) -- ++ (0, 0.5);
	\draw[dashed] (bmid1.south east) -- ++ (0, -0.5);
	
	\draw[-, decoration={markings, mark=at position .5  with {\draw [black, thick,-] 
			++ (-0.3,-0.3) 
			-- (0.3,0.3);}}, postaction=decorate] (bmid1.north) -- node[left, xshift=-5mm, yshift=-2mm]{$(a,k+1)$-split} (amid1.south);
\end{tikzpicture}
	\end{subfigure}%
	\vspace{1em}
	\begin{subfigure}{0.5\linewidth}
		\centering
		\tikzset{ns/.style={
		inner sep = 0mm, outer sep=0mm,
		append after command={
			[very thick]
			(\tikzlastnode.north west)edge(\tikzlastnode.north east)
			[very thick]
			(\tikzlastnode.south west)edge(\tikzlastnode.south east)
		}
	}
}

\tikzset{wens/.style={
		inner sep = 0mm, outer sep=0mm,
		append after command={
			[very thick]
			(\tikzlastnode.north west)edge(\tikzlastnode.south west)
			[very thick]
			(\tikzlastnode.north east)edge(\tikzlastnode.south east)
			[very thick]
			(\tikzlastnode.north west)edge(\tikzlastnode.north east)
			[very thick]
			(\tikzlastnode.south west)edge(\tikzlastnode.south east)
		}
	}
}

\tikzset{wns/.style={
		inner sep = 0mm, outer sep=0mm,
		append after command={
			[very thick]
			(\tikzlastnode.north west)edge(\tikzlastnode.south west)
			[very thick]
			(\tikzlastnode.north west)edge(\tikzlastnode.north east)
			[very thick]
			(\tikzlastnode.south west)edge(\tikzlastnode.south east)
		}
	}
}

\tikzset{ens/.style={
		inner sep = 0mm, outer sep=0mm,
		append after command={
			[very thick]
			(\tikzlastnode.north east)edge(\tikzlastnode.south east)
			[very thick]
			(\tikzlastnode.north west)edge(\tikzlastnode.north east)
			[very thick]
			(\tikzlastnode.south west)edge(\tikzlastnode.south east)
		}
	}
}

\begin{tikzpicture}
	\LARGE
	\node[circle] (t0) at (0,0) {$t$};
	\node[wns, minimum height=6mm, minimum width=15mm, align=center, right= 0mm of t0] (t1) {$\ldots$};
	\node[wens, fill=black!50, minimum height=6mm, minimum width=6mm, align=center, right= 0mm of t1] (t2) {$x$};
	\node[wns, minimum height=6mm, minimum width=18mm, align=center, right= 0mm of t2] (t3) {$\ldots$};
	\node[ns, minimum height=6mm, minimum width=2mm, align=center, right= 0mm of t3] (t4) {$ $};
	\node[wens, fill=black!20, minimum height=6mm, minimum width=24mm, align=center, right= 0mm of t4] (t5) {$ $};
	\node[wens, fill=black!50, minimum height=6mm, minimum width=6mm, align=center, right= 0mm of t5] (t6) {$x$};
	\node[wens, fill=black!20, minimum height=6mm, minimum width=10mm, align=center, right= 0mm of t6] (t7) {$ $};
	\node[wens, minimum height=6mm, minimum width=11mm, align=center, right= 0mm of t7] (t8) {$\ldots$};
	\node[circle, align=center, left= 10mm of t0] (t9) {$0$};
	\node[circle, align=center, above= 6mm of t9] (t10) {$k$};
	\node[circle, align=center, above= -5mm of t10] (t11) {$k+1$};
	
	\node[circle] (s0) at (0,7) {$s$};
	\node[wens, minimum height=6mm, minimum width=15mm, align=center, right= 0mm of s0] (s1) {$\ldots$};
	\node[wens, fill=black!50, minimum height=6mm, minimum width=6mm, align=center, right= 0mm of s1] (s2) {$x$};
	\node[ns, minimum height=6mm, minimum width=3mm, align=center, right= 0mm of s2] (s3) {$ $};
	\node[ens, minimum height=6mm, minimum width=20mm, align=center, right= 0mm of s3] (s4) {$\ldots$};
	\node[wens, fill=black!20, minimum height=6mm, minimum width=16mm, align=center, right= 0mm of s4] (s5) {$ $};
	\node[wens, fill=black!50, minimum height=6mm, minimum width=6mm, align=center, right= 0mm of s5] (s6) {$x$};
	\node[wens, fill=black!20, minimum height=6mm, minimum width=18mm, align=center, right= 0mm of s6] (s7) {$ $};
	\node[wens, minimum height=6mm, minimum width=8mm, align=center, right= 0mm of s7] (s8) {$\ldots$};
	\node[circle, align=center, left= 10mm of s0] (s9) {$0$};
	\node[circle, align=center, below= 3mm of s9] (s10) {$k$};
	\node[circle, align=center, below= -5mm of s10] (s11) {$k+1$};
	
	\node[above=0mm of s2] (prevs1) {$prev(a,x)$}; 
	\node[below=0mm of t2] (prevs2) {$prev(b,x)$}; 
	\node[above=0mm of s6] (right1) {$right(a,x)$};
	\node[below=0mm of t6] (right1) {$right(b,x)$};
	
	\node[wens, minimum height=6mm, minimum width=40mm, xshift=12mm, align=center, below= 10mm of s5] (s5aux) {$a$};
	
	\draw[dashed] (s5aux.north west) edge  (s5.south west);
	\draw[dashed] (s5aux.north east) edge  (s7.south east);
	
	\node[wens, minimum height=6mm, minimum width=40mm, xshift=8mm, align=center, above= 10mm of t5] (b) {$b$};
	
	\draw[dashed] (t5.north west) edge  (b.south west);
	\draw[dashed] (t7.north east) edge  (b.south east);
	
	\draw[-, decoration={markings, mark=at position .5  with {\draw [black, thick,-] 
			++ (-0.3,-0.3) 
			-- (0.3,0.3);}}, postaction=decorate] (s5aux.south east) -- node[right, xshift=2mm]{$k$-split} (b.north east);
	
	\node[wens, minimum height=8mm, minimum width=8mm, align=center, below= 18mm of s4,] (amid1) {$a^\prime$};
	\node[wens, minimum height=8mm, minimum width=8mm, align=center, xshift=-13mm, above= 18mm of t4,] (bmid1) {$b^\prime$};
	
	\draw[dashed] (amid1.north east) -- ++ (0, 1.8);
	\draw[dashed] (bmid1.south east) -- ++ (0, -1.8);
	
	\draw[-, decoration={markings, mark=at position .5  with {\draw [black, thick,-] 
			++ (-0.3,-0.3) 
			-- (0.3,0.3);}}, postaction=decorate] (bmid1.north) -- node[right, xshift=2mm]{$(a,k+1)$-split} (amid1.south);
\end{tikzpicture}
	\end{subfigure}
	\caption{Illustration of the three cases of \cref{lem:aksplit}.}
	\label{fig:lem-aksplit}
\end{figure}

In \cref{lem:aksplit},
because $k\geq 1$ and the blocks $a'$ and $b'$ are the (explicit or implicit) children of two S-connected $k$-blocks,
it follows that $\alp(s[n_{a'}:n])=\alp(t[n_{b'}:n'])$.
This means, in particular,
that $\nextpos(a',x)\neq \infty$ for some letter $x$
if and only if $\nextpos(b',x)\neq \infty$.

Now, we can explain how to algorithmically apply \cref{lem:aksplit} and find the pairs of $(k+1)$-blocks which should be split.
For this, we can define, and compute in the step where the $k$-split pairs were obtained,
a list $L_k$ of pairs of singleton-$k$-blocks which were $k$-split and a list $H_k$ of all the explicit $k$-nodes of $T_s$ and their $P$-connections.

We first consider each explicit $k$-node $a$ of $T_s$ and its P-Connection,
the node $b$ of $T_t$
(in both cases: when $a$ and $b$ were $k$-split or when they were not).
For $x\in \alp(s[m_a-1:n_a-1])\cup \{s[m_a-1]\}$
(note that the symbols $x\in \alp(s[m_a-1:n_a-1])$ can be identified as the first symbols of the $(k+1)$-blocks
into which $a=s[m_a:n_a]$ is split, except the rightmost one; these are the children of node $a$ except the rightmost one)
we do the following:
\begin{enumerate}
	\item identify each $(k+1)$-block $a'=[m_{a'}:n_{a'}]$ with $\prevpos(a,x) < n_{a'}< m_a$
	and its pair $b'=[m_{b'}:n_{b'}]$.
	Then $(a',b')$ is not in the S-Connection
	if $n_{b'} \leq \prevpos(b,x)$
	(i.e., $a'$ and $b'$ are $(a,k+1)$-split). 
	\item identify each $(k+1)$-block $a'=[m_{a'}:n_{a'}]$ with $\prevpos(a,x) < n_{a'}\leq \rightpos(a,x)$
	and its pair $b'=[m_{b'}:n_{b'}]$.
	Then $(a',b')$ is not in the S-Connection
	if $\rightpos(b,x)< n_{b'} \leq n_b$. 
	\item if $a\nsim_k b$, identify each $(k+1)$-block $a'$ with $\prevpos(a,x) < n_{a'}\leq \rightpos(a,x)$
	and its pair $b'=[m_{b'}:n_{b'}]$.
	Then $(a',b')$ is not in the S-Connection
	if $\prevpos(b,x)< n_{b'} \leq \rightpos(b,x)$. 
\end{enumerate}

For every pair $(a,b)$ of singleton-$k$-blocks
which were $k$-split (from the list $L_k$),
we only perform step 3 from above. 

For each $k$-block $a$ we considered
(explicit or implicit node of $T_s$),
we collect the singleton-$(k+1)$-blocks that were $(a,k+1)$-split,
to be used when computing the $(k+2)$-splits.

The next step is to implement this idea,
i.e., to describe data structures allowing us to identify efficiently the $(k+1)$-blocks $a'$ and $b'$ from above.
We say that a pair of blocks/ nodes $(a',b')$ meets an interval-pair $([p:q],[p':q'])$
if $a'$ ends in $[p:q]$, and $b'$ ends in~$[p':q']$. 

Our approach is the following.
We process the blocks on level $k$
and, for each of them, get (at most) three lists of interval-pairs
(one component is an interval of positions in $s$, the other an interval in $t$).
On level $k+1$, we split each pair of P-connected blocks $(a',b')$
which meets one interval-pair from our list.
A crucial property here is that,
for each interval-pair,
the $(k+1)$-blocks of $s$ which meet it,
and are accordingly split from their P-Connections,
are consecutive (explicit and implicit) $(k+1)$-nodes in $T_s$.
Thus, in order to make use of \cref{lem:aksplit},
we draw on the technical results given by \cref{lem:init,lem:prev,comput_intervals,lem:level1},
which we collected in the following section for the interested reader.\looseness=-1

\subsection{Technical Tools}
\begin{restatable}{lemma}{restateInit}\label{lem:init}
	Given two words $s$ and $t$, with $|s|=n$ and $|t|=n'$, $n\geq n'$,
	and their Simon-Trees $T_s$ and $T_t$,
	we can process the trees $T_s$ and $T_t$ in $O(n+n')$ time
	such that the following information can be retrieved in $O(1)$ time:
	\begin{enumerate}
		\item For an (explicit or implicit) node $a$ of $T_s$,
		the (explicit or implicit) node $b$ of $T_t$ to which it is P-connected.
		\item The $j^{th}$ (from left to right) explicit node on level $k$ of $T_s$ (respectively, $T_t$).
		Note that, because $T_s$ and $T_t$ are ordered trees,
		we can uniquely identify the $j^{th}$ (from left to right) explicit node
		that occurs on level $k$. 
		\item For each position $i$ of $s$,
		the unique position $U[i]$ of $t$ such that the singleton node $[i:i]$ is P-connected to the singleton node $[U[i]:U[i]]$. 
		\item For each position $i$ of $t$,
		the unique position $U'[i]$ of $s$ such that the singleton node $[i:i]$ is P-connected to the singleton node $[U'[i]:U'[i]]$. 
		\item For each position $i$ of $s$ (respectively of $t$) and level $k$,
		the node associated to the $k$-block that starts with $i$,
		if such a node exists.
	\end{enumerate}
\end{restatable}
\begin{proof}
	We can directly use \cref{def:pcon} to compute the P-Connection on levels.
	The roots of the trees are P-connected.
	Then, we start the traversal of the trees and,
	when considering a pair of P-connected nodes,
	we connect their respective $i^{th}$ children, for all $i$.
	The only aspect that needs to be treated carefully is that each time we place a pair of explicit nodes $([i:i], [j:j])$ in the P-Connection,
	we can set $U[i]=j$ and $U'[j]=i$,
	and note that the (implicit or explicit) nodes $[i:i]$ and $[j:j]$ will be P-connected on all lower levels.
	This solves items $1$, $3$, $4$ of the list in the statement.
	For simplicity, we can also assume that the P-Connection is implemented as a series of arrays $P_k$ and $P'_k$
	(where $k$ is a level of $T_s$)
	such that  $P_k[j]=j'$ and $P'_k[j']=j$
	if and only if the $j^{th}$ node on level $k$ of $T_s$
	(from left to right)
	is connected to the $j'^{th}$ node on level $k$ of $T_t$
	(also from left to right). 
	
	The rest of the proof is given for $s$ and $T_s$.
	An analogous approach works for $t$ and $T_t$. 
	
	By traversing the tree $T_s$ on levels (left to right)
	we can also associate to each explicit node of the tree the pair $(k,j)$,
	where $k$ is its level,
	and $j-1$ is how many explicit nodes occur to the left of that node on its level.
	As such, we can construct an array  for each level of the tree, namely $s_k$,
	such that $s_k[j]$ is the explicit node associated to the pair $(k,j)$ from $T_s$.
	For $t$ and $T_t$ we construct the array $t_k$ for each level of the tree. 
	
	This solves item 2 of the above list in linear time.
	
	During the traversal of $T_s$,
	we can also compute for each level $k$ an array~$N_{s,k} $
	such that if $[p:q]$ is a $k$-block of $T_s$,
	then $N_{s,k}[p]$ is a pointer to the node associated to the block $[p:q]$ of $T_s$.
	In $T_t$, we define the arrays $N_{t,k}$,
	where $N_{t,k}[p]$ points to the node associated to $[p:q]$ from $T_t$. 
	
	So, the whole process takes linear time,
	and all the desired information can be retrieved in $O(1)$ using the arrays we constructed.
	The statement follows.
\end{proof}

\begin{restatable}{lemma}{restatePrev}\label{lem:prev}
	Given two words $s$ and $t$, with $|s|=n$ and $|t|=n'$, $n\geq n'$,
	and their Simon-Trees $T_s$ and $T_t$, respectively,
	we can compute in $O(n)$ time all the values:
	\begin{itemize}
		\item $\prevpos(a,x)$ and $\rightpos(a,x)$
		for $a=[m_a:n_a]$ a block in $T_s$ and $x\in \alp(a)\cup\{s[m_a-1]\}$. 
		\item $\prevpos(b,x)$ and $\rightpos(b,x)$
		for $b=[m_b:n_b]$ a block in $T_t$,
		which is $P$-connected to the block $a=[m_a:n_a]$ of $T_s$,
		and $x\in \alp(a)\cup\{s[m_a-1]\}$. 
	\end{itemize}
\end{restatable}
\begin{proof}
	The first observation is that $\rightpos(a,x)$ for $a=[m_a:n_a]$ a $k$-block in $T_s$ and $x\in \alp(a)\cup\{s[m_a-1]\}$ can actually be computed by looking at the splitting of the block $a$ into $(k+1)$-blocks.
	By \cref{lem:split},
	each position $\rightpos(a,x)$,
	with $x\in \alp(a)\cup\{s[m_a-1]\}$,
	is a position where one such $(k+1)$-block starts.
	Similarly, we can compute $\rightpos(b,x)$ for $b=[m_b:n_b]$ a $k$-block in $T_t$,
	which is $P$-connected to the block $a=[m_a:n_a]$ of $T_s$,
	and $x\in \alp(a)\cup\{s[m_a-1]\}$.
	Again, these positions are among the positions that split $b$ into $(k+1)$-blocks. 
	
	To efficiently manage the P-Connection,
	we can first compute it using \cref{lem:init}.
	To retrieve $\rightpos(a,x)$ in $O(1)$ time,
	we store it in the child of $a$ representing a block whose last letter is $x$. 
	
	To compute all the values $\prevpos(a,x)$ for all blocks $a$ of $T_s$
	and $x\in \alp(a)\cup\{s[m_a-1]\}$,
	we do the following:
	\begin{enumerate}
		\item We represent the query $\prevpos(a,x)$ for a block $a=[m_a:n_a]$ of $T_s$ and $x\in \alp(a)\cup\{s[m_a-1]\}$ as the triple $(x,m_a,n_a)$.
		For each such triple, we store a pointer to the node $a$. 
		\item We radix-sort the triples $\{(x,m_a,n_a)\mid a=[m_a:n_a]$ is a block of $T_s$ and $x\in \alp(a)\cup\{s[m_a-1]\}\}$
		and obtain a list $L$. 
		\item For each $x\in \Sigma$,
		we select in a list $L_x$ the contiguous part of $L$
		consisting in all the triples of the form $(x,i,i')$. 
		\item We initialize an array $Q$ with $\len\Sigma$ elements,
		where $Q[x]=0$ for all $x\in \Sigma$. 
		\item We now go through the letters $s[j]$ of the word $s$, for $j$ from $1$ to $n$.
		If $s[j]=x$ holds,
		then we remove from the list $L_x$ all the triples $(x,i,i')$ with $i-1 \leq j$;
		for each such triple $(x,i,i')$ and the block $a=[i:i']$,
		we set $\prevpos(a,x)=Q[x]$.
		Then, we set $Q[x]=j$.  
	\end{enumerate}
	It is immediate that the above algorithm computes correctly $\prevpos(a,x)$  for all blocks $a$ of $T_s$ and $x\in \alp(a)\cup\{s[m_a-1]\}$.
	In order to access these values efficiently,
	we store $\prevpos(a,x)$ in the children of $a$ ending with $x$, for $x\in \alp(a)$,
	and as a separate satellite value in the node $a$ for $x=s[m_a-1]$.
	The complexity is clearly linear. 
	
	A similar algorithm can be used to compute the values $\prevpos(b,x)$ for $b=[m_b:n_b]$ a block in $T_t$,
	which is $P$-connected to the block $a=[m_a:n_a]$ of $T_s$,
	and $x\in \alp(a)\cup\{s[m_a-1]\}$. 
\end{proof}

\begin{restatable}{lemma}{restateComputIntervals}\label{comput_intervals}
	Let $a=[m_a:n_a]$ and $b=[m_b:n_b]$ be P-connected blocks of $T_s$ and $T_t$, respectively,
	and $s_a=s[m_a-1:n_a-1]$.
	We can compute in overall $O(|\alp(a)|)$ time the three lists,
	associated to the pair $(a,b)$, containing:
	\begin{enumerate}
		\item the interval-pairs $([\prevpos(a,x)+1: m_a-1],[0: \prevpos(b,x)])$,
		for all $x\in \alp(s_a)$; 
		\item the interval-pairs $([\prevpos(a,x)+1:\rightpos(a,x)],[\rightpos(b,x)+1: n_b])$,
		for all $x\in \alp(s_a)$; 
		\item the interval-pairs $([\prevpos(a,x)+1:\rightpos(a,x)],[\prevpos(b,x)+1: \rightpos(b,x)])$,
		for all $x\in \alp(s_a)$.
	\end{enumerate}
\end{restatable}
\begin{proof}
	We first run the algorithm of \cref{lem:prev}
	and then the required interval-pairs can be clearly computed in $O(1)$ time per pair.
	This adds up to $O(|\alp(a)|)$ time for a pair of P-connected nodes $(a,b)$. 
\end{proof}

\begin{restatable}{lemma}{restateLevel}\label{lem:level1}
	Given two words $s$ and $t$, with $|s|=n$ and $|t|=n'$, $n\geq n'$, and their Simon-Trees $T_s$ and $T_t$,
	we can check in $O(n)$ overall time for all pairs of P-connected $1$-blocks $(a,b)$,
	with $a=[m_a:n_a]$ and $b=[m_b:n_b]$,
	whether $\alp(s[n_a:n])=\alp(t[n_b:n'])$. 
\end{restatable}
\begin{proof}
	Let $\sigma=|\Sigma|$ be the size of the input alphabet.
	Clearly, we have $\sigma\in O(n)$
	and we can assume that $\Sigma = [1:\sigma]$. 
	
	Let $g$ be the number of $1$-nodes in $T_s$
	(recall that these nodes are numbered from right to left),
	and $g'$ be the number of $1$-nodes of $T_t$.
	The partition of $s$ (resp., of $t$) into $1$-blocks is done according to \cref{lem:split},
	and it can be retrieved from the first level of the tree $T_s$ (and $T_t$ respectively).
	Indeed, we split the interval $[1:n]$ into the intervals/blocks corresponding to the children of the $0$-node of $T_s$.
	We also apply these splits in the split-find structure ${\mathcal S}$. 
	
	Let us see now how to synchronise the blocks of the two trees.
	We use an array $H$ with $\sigma$ elements,
	initially all set to $0$.
	As an invariant,
	when processing the $i^{th}$ block $a=[m_a:n_a]$ from $T_s$
	and the $i^{th}$ block $b=[m_b:n_b]$ from $T_t$ we have that,
	for each letter $x\in \Sigma$: 
	\begin{itemize}
		\item $H[x]=0$ if $x \notin \alp(s[n_a:n]) \cup  \alp(t[n_b:n'])$,  
		\item $H[x]=1$ if $x \in \alp(s[n_a:n]) \setminus  \alp(t[n_b:n'])$,  
		\item $H[x]=2$ if $x \in \alp(s[n_a:n]) \cap  \alp(t[n_b:n'])$,  
		\item $H[x]=3$ if $x \in  \alp(t[n_b:n']) \setminus \alp(s[n_a:n])$.
	\end{itemize}
	Also, we maintain a variable $check$ that counts how many odd values are in $H$.
	Clearly, $\alp(s[n_a:n])=\alp(t[n_b:n'])$
	if and only if $check$ is $0$.
	
	We maintain $check$ and $H$ as follows.
	For $i$ from $1$ to $g$,
	assume that the $i^{th}$ node on level $1$ of $T_s$ is $a=[m_a:n_a]$,
	and the $i^{th}$ node on level $1$ of $T_t$ is $b=[m_b:n_b]$. 
	If $H[s[n_a]]=0$,
	then set $H[s[n_a]]\gets 1$,
	and $check$ is increased by $1$.
	If $H[s[n_a]]=3$,
	then set $H[s[n_a]]\gets 2$,
	and $check$ is decreased by $1$.
	If $H[t[n_b]]=0$,
	then set $H[t[n_b]]\gets 3$,
	and $check$ is increased by $1$.
	If $H[t[n_b]]=1$,
	then set $H[t[n_b]]\gets 2$,
	and $check$ is decreased by $1$.
	In all other cases we do nothing.
	After this, if $check$ is odd,
	then $\alp(s[n_a:n])\neq \alp(t[n_b:n'])$.
	Otherwise, $\alp(s[n_a:n])=\alp(t[n_b:n'])$.
	We then consider the next value of $i$.
	
	The algorithm described above can be clearly implemented in $O(n+n')$ time.
\end{proof}

\section{Efficiently constructing the S-Connection and solving \maxkproblem.}
\label{sec:s-connection-algo}
Based on the previous lemmas,
we can now show our main technical theorem.
We use \cref{lem:level1} to see which $1$-nodes are not S-connected.
This is done in $O(n)$ time.
Then consider the $k$-nodes,
for each $k\geq 2$ in increasing order.
For each pair $(a,b)$ of $(k-1)$-nodes which were split
(i.e., removed from the S-Connection)
in the previous step,
we split the pairs of $k$-nodes meeting one of the interval-pairs of the three lists of $(a,b)$,
as computed in \cref{comput_intervals}.
To do this efficiently,
we maintain an interval union-find and an interval split-find structure for each word.

While the concrete algorithm can be found in the following \cref{sec:s-connection-algo-descr} of this paper,
the main result is stated in \cref{sec:main-result}.\looseness=-1

\subsection{S-Connection construction algorithm}
\label{sec:s-connection-algo-descr}
Let $\sigma=|\Sigma|$ be the size of the input alphabet.
Clearly, we have $\sigma\in O(n)$,
and we can assume that $\Sigma = [1:\sigma]$.
We first present our algorithm and then show that it fulfils the desired properties.

\subparagraph*{Data structures and preprocessing.}
We maintain an interval union-find data structure $\mathcal{U}$
and two interval split-find data structures:
${\mathcal S}_s$ over the universe $[1:n]$ and ${\mathcal S}_t$ over the universe $[1:n']$.
Initially, $\mathcal{U}$ contains all the separate intervals $[i:i]$ for $i\in [1:n]$,
while ${\mathcal S}_s$ (respectively, ${\mathcal S}_t$) contains just the interval $[1:n]$ (the interval $[1:n']$).
We assume that the $\find_{\mathcal U}(x)$
(respectively, $\find_{{\mathcal S}_s}(x)$ and $\find_{{\mathcal S}_t}(x)$)
operation returns the borders of the interval of ${\mathcal U}$
(respectively, ${\mathcal S_s}$ and ${\mathcal S}_t$)
which contains $x$.
We will use ${\mathcal U}$ to keep track of the positions of $s$
which were split from the positions to which they are P-connected,
while ${\mathcal S}_s$ and ${\mathcal S}_t$ are used to maintain the splitting of $s$
and, respectively, $t$ into $k$-blocks while $k$ is incremented during the algorithm. 

We also maintain an array with $n$ components $Level_s$ and an array with $n'$ components $Level_t$.
Initially, all components are set to $\infty$.
At the end of the computation,
these arrays will store,
for each position $i$ of $s$ (respectively, of $t$),
the highest level on which $i$ is contained in an implicit or explicit block of $s$
(respectively, $t$)
which is not S-connected to any node of $T_t$
(respectively, $T_s$).

In this initial phase,
we compute the Simon-Trees $T_s$ and $T_t$ using the algorithm from \cref{sec:simon-tree-construction}.
We also compute the data structures in \cref{lem:init,lem:prev}.
Finally, we can also compute the lists of interval-pairs from \cref{comput_intervals}
(stored, e.g., as three separate lists of interval-pairs)
for each pair of P-connected nodes.

\subparagraph*{The main algorithm.}
Now we move on to the main phase of the algorithm.
The $0$-nodes of the two trees are clearly S-connected.
In the data structures ${\mathcal S}_s$ and ${\mathcal S}_t$
we split the respective $0$-blocks into the corresponding $1$-blocks
(these are given by their children in $T_s$ and $T_t$).

{\em \S\ The first step} of our algorithm is {\em computing the S-Connection between $1$-nodes}. 

For $k=1$, we process $s$, $t$, and their corresponding Simon-Trees according to \cref{lem:level1}. 

Now, we 1-split the pair $a=[m_a:n_a]$ and $b=[m_b:n_b]$
if and only if $\alp(s[n_a:n])\neq \alp(t[n_b:n'])$.
Further, for all $\ell\in [m_a:n_a]$,
we set $Level_s[\ell]=1$,
as $1$ is the level on which position $\ell$ was split from its P-Connection.
Similarly, for $\ell' \in  [m_b:n_b]$
we set $Level_t[\ell']=1$.
We also update the union-find structure ${\mathcal U}$.
First, we make the union of the singletons $[i:i]$,
for $i$ from $m_a$ to $n_a$.
Then, if $Level_s[m_a-1]\neq \infty$,
we make the union between the interval that contains $m_a-1$
(returned by $\find_S(m_a-1)$)
and the interval that contains $m_a$ (namely, $[m_a:n_a]$).
Further, if $Level_s[n_a+1]\neq \infty$,
then we make the union between the interval that contains $m_a$
(returned by $\find_S(m_a)$)
and the interval that contains $n_a+1$
(returned by $\find_S(n_a+1)$).
In this way, we ensure that each interval of consecutive positions $i$ of $s$,
for which $Level_s[i]\neq \infty$ and which cannot be extended to the left nor to the right,
corresponds to a single interval stored in the interval union-find structure ${\mathcal U}$. 

If $c=[m_c:n_c]$ is a $1$-block of $s$ (or, respectively, $t$)
which was not P-connected to a block of $s$ (respectively, $t$),
then, for $\ell' \in  [m_c:n_c]$,
we set $Level_s[\ell']=1$ (respectively, $Level_t[\ell']=1$).

At the end of this step
we collect in a list $L_1$ the positions $i$ for which the $1$-blocks $[i:i]$ and $[U[i]:U[i]]$ were $1$-split
(so, the pairs of singleton-$1$-blocks which were split).
Finally, we split the intervals $[1:n]$ (and $[1:n']$) in the data structures ${\mathcal S}_s$ and ${\mathcal S}_t$ into the corresponding $2$-blocks
(i.e., each $1$-block is split into its children). 

We note that at the end of this step
we will have that $Level_s[i]\neq \infty $
if and only if the position $i$ of $s$ was contained in a $1$-block of $T_s$
which was split from the block of $t$ to which it is P-connected
(that is, the respective pair of blocks is not in the S-Connection).
Clearly, for any $2$-block $a=[m_a:n_a]$ of $s$,
a position $i \in [m_a:n_a]$ fulfils $Level_s[i]\neq \infty$
if and only if all positions $j\in [m_a:n_a]$ fulfil $Level_s[j]\neq \infty$.  

The pairs of P-connected $1$-blocks which were not split, are S-connected.

{\em \S\  The iterated step} of our algorithm is {\em computing the S-Connection between $k+1$-nodes}, for $k$ from $1$ to $d-1$,
where $d$ is the last level of $T_s$. 

We can assume that we have the list $L_k$ containing the pairs of (explicit and implicit) singleton-$k$-blocks
which were $k$-split (that is, split in the previous iteration),
as well as the list $H_k$ of explicit nodes on level $k$ of $T_s$
paired with the nodes from level $k$ of $T_t$ to which they are P-connected. 

So let $(a,b)$ be a pair of nodes from $L_k \cup H_k$.
That is, $a=[m_a:n_a]$ is a $k$-block from $T_s$ (explicit or implicit),
$b=[m_b:n_b]$ is a $k$-block from $T_t$,
and $a$ and $b$ are P-connected.
Using \cref{comput_intervals}, in the preprocessing phase,
we have computed the following three lists of pairs of intervals:
\begin{enumerate}
	\item[1.] For each letter $x$ occurring in $s[m_a-1:n_a-1]$,
	the interval-pair $([\prevpos(a,x)+1: m_a-1],[0: \prevpos(b,x)])$.
	\item[2.] For each letter $x$ occurring in $s[m_a-1:n_a-1]$,
	the interval-pair $([\prevpos(a,x)+1:\rightpos(a,x)],[\rightpos(b,x)+1: n_b])$. 
	\item[3.] For each letter $x$ occurring in $s[m_a-1:n_a-1]$, 
	the interval-pair $([\prevpos(a,x)+1:\rightpos(a,x)],[\prevpos(b,x)+1: \rightpos(b,x)])$. 
\end{enumerate}

Let $([e_1:e_2],[f_1:f_2])$ be an interval-pair contained in one of the first two lists.
We should $(a,k+1)$-split all the pairs of P-connected $(k+1)$-blocks $(a',b')$ (explicit or implicit)
such that $a'$ ends inside $[e_1:e_2]$ and $b'$ inside $[f_1:f_2]$.
For this, we will use our additional data structures ${\mathcal U}$ and ${\mathcal S}$. 

This is done as follows.
Using $\find_{{\mathcal S}_s}(e_1)$ and $\find_{{\mathcal S}_t}(f_1)$,
we compute the $(k+1)$-block $[p,q]$ of $s$ that contains $e_1$
and, respectively, the $(k+1)$-block $[p',q']$ of $t$ that contains $f_1$.
Let $[r:o]$ be the $(k+1)$-block of $s$
which is P-connected to $[p',q']$.
We can obtain $[r:o]$ according to \cref{lem:init},
which we already used in the preprocessing phase.
If $r\geq p$,
then we set $p\gets r$ and $q\gets o$
to ensure that $[p:q]$ is the leftmost $(k+1)$-block of $s$ that ends to the right or on $e_1$,
which is P-connected to a block which ends to right or on $f_1$. 

Further, if $Level_s[p]\neq \infty$,
then we use $\find_{\mathcal U}(s)$ to identify the interval $[d_1:d_2]$ of positions of $s$
which contains $p$ and has the property
that if $i\in [d_1:d_2]$,
then $Level_s[i]\neq \infty$.
We then set $p\gets d_2+1$ and $q$ as the right border of the interval $\find_S(p)$
(i.e., the right border of the $(k+1)$-block starting with $p$).
Now, we have $Level_s[p]=\infty$,
and we can continue as described below. 

The main part of the processing we do for $k+1$ consists in the following {\em loop}. 

If $Level_s[p]\neq \infty$,
then we check whether the $(k+1)$-block $[p:q]$ ends in $[e_1:e_2]$.
If yes, we compute the block $[p':q']$ to which it is P-connected
and see if it ends inside $[f_1:f_2]$.
If any of the previous checks is false,
then we stop the execution of the {\em loop}.
If both checks are true,
then we decide that $[p:q]$ and $[p':q']$ are not S-connected
(the reason for this being that $s[\nextpos([p:q],x)+1:n]\nsim_k t[\nextpos([p':q'],x)+1:n']$).
Thus, we set $Level_s[i]=k+1$ for all $i\in [p:q]$ and $Level_t[i]=k+1$ for all $i\in [p':q']$.
We also make in ${\mathcal U}$ the union of all the intervals $[i:i]$ with $i\in [p:q]$.
Then we make the union of the interval $[p:q]$ and the interval of $p-1$
if $Level_s[p-1]\neq \infty$,
and the union of the interval of $q$ and the interval of $q+1$
if $Level_s[q+1]\neq \infty$. 

Then, we consider the $(k+1)$-block starting with $p+1$,
returned by $\find_{\mathcal S}(p+1)$,
and repeat the {\em loop} for this block in the role of $[p:q]$.

The {\em loop} stops when we reached a $(k+1)$-block $[p:q]$
that ends outside of $[e_1:e_2]$
or it is P-connected to a block $[p':q']$
that ends outside $[f_1:f_2]$.
In both cases, we consider a new interval-pair $([e_1:e_2],[f_1:f_2])$ from our lists
and repeat the same process described above.

Moving on to the third list,
we will process it exactly as the first two lists,
but only in the case when it corresponds to a pair $(a,b)$ of $k$-blocks
that are not S-connected. 

Finally, if $a=[m_a:n_a]$ is a block of $s$ (respectively, of $t$)
which did not have a pair in the P-Connection (due to \cref{lem:pcon}),
we set $Level_s[i]=k+1$ (respectively, $Level_t[i]=k+1$) for all $i \in [m_a:n_a]$.
If $a$ is a block of $s$,
we also make in ${\mathcal U}$ the union of all the intervals $[i:i]$ with $i\in [m_a:n_a]$.
Then we make the union of the interval $[m_a:n_a]$ and the interval of $m_a-1$
if $Level_s[m_a-1]\neq \infty$,
and the union of the interval of $n_a$ and the interval of $n_a+1$
if $Level_s[n_a+1]\neq \infty$.
Once more, we ensure that each interval of consecutive positions $i$ of $s$,
for which $Level_s[i]\neq \infty$
and which cannot be extended to the left nor to the right,
corresponds to a single interval stored in the interval union-find structure ${\mathcal U}$. 

Similarly to {\em the first step}, at the end of this step,
we collect in a list $L_{k+1}$ the positions $i$ for which the P-Connection between the $(k+1)$-blocks $[i:i]$ and $[U[i]:U[i]]$ were split,
and we split the non-singleton $k$-blocks of $s$ and $t$,
and the corresponding intervals of ${\mathcal S}_s$ and ${\mathcal S}_t$,
into their $(k+1)$-blocks
(i.e., each $k$-block is split into its children). 

As an invariant, at the end of the execution of the iteration for $k$,
we will have that $Level_s[i]=j\neq \infty $
if and only if the position $i$ of $s$ was contained in an explicit or implicit $(k+1)$-block of $T_s$
which was split in the iteration $j-1$ from the block of $T_t$ to which it is P-connected. Clearly, for any $g$-block $a=[m_a:n_a]$ of $s$, with $g> k+1$,
a position $i \in [m_a:n_a]$ fulfils $Level_s[i]\neq \infty$
if and only if all positions $j\in [m_a:n_a]$ fulfil $Level_s[j]\neq \infty$. 

The pairs of P-connected $(k+1)$-blocks (implicit or explicit)
which were not split,
are S-connected.

\subsection{Main result}
\label{sec:main-result}
After the complete algorithm description,
we can now state our main result in \cref{thm:main}.\looseness=-1

\begin{restatable}{theorem}{restateMain}\label{thm:main}
	Given two words $s$ and $t$, with $|s|=n$ and $|t|=n'$, $n\geq n'$,
	we can compute in $O(n)$ time the following:
	\begin{itemize}
		\item the S-Connection between the nodes of the two trees $T_s$ and $T_t$;
		\item for each $i\in [1:n]$,
		the highest level $k$ on which the (implicit or explicit) node $[i:i]$ is $k$-split from its P-Connection. 
	\end{itemize}
\end{restatable}
\begin{proof}
	\begin{description}
		\item[The output of the algorithm and its correctness.]
		Our algorithm from \cref{sec:s-connection-algo-descr} outputs the S-Connection we computed
		as well as the arrays $Level_s$ and $Level_t$;
		these store for each position $i$ (of $s$ or $t$) the highest level on which the position $i$ is part of a block in its respective tree
		which is not S-connected to any node of the other tree.
		The correctness of our algorithm follows from \cref{lem:2words-block-equivalence,lem:pcon,lem:aksplit}:
		we split all pairs of P-connected nodes which are not S-connected.
		Moreover, we do this considering the nodes of the trees in increasing order of their levels,
		which proves that the arrays $Level_s$ and $Level_t$ are correctly computed.
		\item[The complexity of the algorithm.]
		According to \cref{union-find},
		we can assume that constructing the data structures ${\mathcal U}$ and ${\mathcal S}$ takes linear time,
		and the time needed to execute all the operation on these structures is linear in their number,
		i.e., each operation takes $O(1)$ amortized time.
		Therefore, according to \cref{lem:init,lem:prev,comput_intervals},
		the {\em preprocessing phase} takes linear time.
		In particular, by \cref{comput_intervals},
		we need $O(|\alp(a)|)$ time to compute the interval-pair lists for two P-connected nodes $a$ and $b$;
		this is proportional to the number of children of $a$,
		so summing this up over all pairs results in a time proportional to the number of nodes of $T_s$, so $O(n)$ time.
		The running time of both the {\em first step} and each of the iterations of the {\em iterated step} is linear in the sum of the number of positions $i$ of $s$ and $t$
		for which we set $Level_s[i]$ (respectively, $Level_t[i]$) to a value different from $\infty$ and the total number of elements contained in the three lists for the nodes of the lists $L_k$ and $H_k$, for all $k$.
		Because for each position $i$ we change the value of $Level_s[i]$ (or $Level_t[i]$) exactly once from $\infty$ to some $k$, 
		then the number of positions $i$ of $s$ and $t$,
		for which we set $Level_s[i]$ (respectively, $Level_t[i]$) to a value different from $\infty$, is $O(n+n')$.
		Then, for each pair $(a,b)$ of $H_k$,
		the total number of interval-pairs in the three lists associated to $(a,b)$
		is linear in the number of children of the node $a$.
		Summing up for all explicit nodes $a$ of $T_s$
		the total number of interval-pairs occurring in the lists associated to $(a,b)$
		(where $a$ and $b$ are P-connected)
		is clearly linear in the number of nodes of $T_s$
		(as each child is counted once).
		Then, for each pair $(a,b)$ of $L_k$,
		the total number of interval-pairs in the three lists is $O(1)$.
		Thus, the total number of elements contained in the three lists for the nodes of the lists $L_k$ and $H_k$, for all $k$, is $O(n)$.
		
		We can now conclude that the overall time needed to perform the {\em first step} and all {\em iterated steps} is $O(n+n')$.
		Thus, the whole algorithm runs in linear time,
		so the statement of the theorem is correct.\looseness=-1
	\end{description}
\end{proof}

Finally, in order to solve \maxkproblem,
we need to compute the largest $k$
for which the $k$-block $a=[1:n_a]$ of $s$ is S-connected to the $k$-block $b=[1:n_b]$ of $t$.
Thus, we execute the algorithm of \cref{sec:s-connection-algo}
and the aforementioned level $k$ can be easily found by checking,
level by level,
the blocks that contain position $1$ of $s$ on each level of $T_s$
and the block to which they are S-connected in $T_t$.
As a consequence of \cref{thm:main},
we can now show our main result.\looseness=-1

\begin{restatable}{theorem}{restateMainOne}\label{thm:main1}
	Given two words $s$ and $t$, with $|s|=n$ and $|t|=n'$, $n\geq n'$,
	we can solve \maxkproblem\ and compute a distinguishing word of minimum length for $s$ and $t$ in $O(n)$ time. 
\end{restatable}
\begin{proof}
	To find a distinguishing word of minimum length we proceed as follows.
	
	We first run the algorithm of \cref{sec:s-connection-algo}
	and store all its additional data structures,
	including those produced in the preprocessing phase.
	We also use an additional array $Y$ with $|\Sigma|$ elements,
	all initialised to $0$.
	We will use this algorithm to show the following general claim.
	
	\subparagraph*{Claim.}
	Let $i$ be a position of $s$ and $j$ be a position of $t$, and let $k> 1$.
	Assume $s[i:n]\nsim_k t[j:n']$,
	but $s[i:n]\sim_{k-1} t[j:n']$,
	and that $i$ is included in the $(k-1)$-block $a=[m_a:n_a]$
	and $j$ is included in the $(k-1)$-block $b=[m_b:n_b]$.
	According to \cref{lem:2words-block-equivalence},
	there exists a letter $x$
	such that $s[\nextpos(i, x) + 1 : n] \nsim_{k-1} t[\nextpos(j, x) +1 : n']$;
	moreover, a word of length $k$ that distinguishes $i$ and $j$ starts with $x$.
	We can find the letter $x$ and its position in time
	$O(|\alp(a)|+|\alp(b)|+(\nextpos(i, x)-i)+(\nextpos(j, x)-j))$. 
	
	Note that $x$ occurs always after both $i$ and $j$ in their respective words,
	or $s[i:n] \nsim_{k-1} t[j:n']$ would hold.
	
	\subparagraph*{Proof of the claim.}	
	Indeed, we can find this $x$ and $\nextpos(i, x)$ and $\nextpos(j, x)$ as follows. 
	
	Assume first that $i$ and $j$ are included, respectively,
	in a pair of P-connected $k$-blocks $a'=[m_{a'}:n_{a'}]$ and $b'=[m_b':n_{b'}]$.
	Then, in the {\em iterated step} of the algorithm from \cref{sec:s-connection-algo},
	we have identified a letter $x$
	such that $s[\nextpos(a',x)+1:n]\nsim_{k-1} t[\nextpos(b',x)+1:n']$;
	this letter can be stored as a satellite information in the $Level_s$ and $Level_t$ arrays.
	It is easy to note that this letter $x$ appears in at most one of the factors $s[i:n_{a'}]$ and $t[j:n_{b'}]$.
	Clearly, in $O((\nextpos(i, x)-i)+(\nextpos(j, x)-j))$
	we can search letter by letter for the first occurrence of $x$ after $i$ and $j$, respectively;
	that is, we compute $\nextpos(i, x)$ and $\nextpos(j, x)$, respectively. 
	
	If $i$ and $j$ are included, respectively,
	in a pair of $k$-blocks $a'=[m_{a'}:n_{a'}], b'=[m_b':n_{b'}]$
	which are not P-connected,
	then there exists a letter $x$ that occurs in $a$ after $i$
	but does not occur in $b$ after $j$, or vice versa.
	This is found as follows:
	we produce the lists $L_a$ and $L_b$ of last letters of the $k$-blocks
	which partition the $(k-1)$ blocks $a$ and $b$, respectively,
	and remove those that occur to the left of $i$ and $j$, respectively.
	In $Y$, we set $Y[y]\gets 1$ for all $y$ that occur in $L_a$.
	Then, we set $Y[y]\gets Y[y]-1$ for all $y\in L_b$.
	If there is a letter $y$ of $L_a$ such that $Y[y]=1$,
	then we choose the letter $x$ which we were searching as $y$.
	Otherwise, there must be a letter $y$ of $L_b$
	such that $Y[y]=-1$,
	and we choose the letter $x$ which we were searching as $y$.
	Furthermore, we reset $Y[y]=0$ for all $y\in L_a\cup L_b$.
	Finding the letter $x$ as above takes $O(|\alp(a)|+|\alp(b)|)$ time.
	Furthermore, finding $\nextpos(i, x)$ and $\nextpos(j, x)$ takes
	$O((\nextpos(i, x)-i)+(\nextpos(j, x)-j))$ time. 
	
	Our claim follows.$\ \triangleleft  $
	
	Now, we show the main statement.
	According to \cref{thm:main},
	we compute a value $k$
	such that  $s[1:n]\nsim_k t[1:n']$,
	but $s[1:n]\sim_{k-1} t[1:n']$.
	Let us assume that $k>1$.
	Using our {\em claim},
	we find a letter $x_1$
	such that $s[\nextpos(1, x_1)+1:n] \nsim_{k-1} t[\nextpos(1, x_1)+1:n']$.
	Then, $s$ and $t$ will be distinguished by a word $w=x_1w'$, with $|w'|=k-1$. 
	This takes $O(|\alp(a_{k-1})|+|\alp(b_{k-1})|+(\nextpos(1, x_1)-1)+(\nextpos(1, x_1)-1))$ time,
	where $a_{k-1}$ is the $(k-1)$-block of $s$ in which $1$ is included,
	and $b_{k-1}$ is the $(k-1)$-block of $t$ in which $1$ is included. 
	
	We continue by searching the first letter $x_2$ of $w'$ as the first letter of a word distinguishing $s[\nextpos(1, x_1) + 1 : n]$ and $ t[\nextpos(1, x_1) +1 : n']$.
	Because we have $s[1:n]\nsim_k t[1:n']$,
	but $s[1:n]\sim_{k-1} t[1:n']$,
	it follows that $s[\nextpos(1, x_1) + 1 : n] \nsim_{k-1} t[\nextpos(1, x_1) +1 : n']$
	but $s[\nextpos(1, x_1) + 1 : n] \sim_{k-2} t[\nextpos(1, x_1) +1 : n']$
	(otherwise, $s[1:n]$ would be distinguished by a word of length $k-1$ from $t[1:n']$). 
	
	This means we can apply the {\em claim} for
	$i=\nextpos(1, x_1)+1$
	(included in the $k-2$ block $a_{k-2}$ of $s$)
	and $j=\nextpos(1, x_1) +1$
	(included in the $k-2$ block $b_{k-2}$ of  $t$)
	and $k-1$ instead of $k$.
	We repeat this until we reach two positions $i$ and $j$,
	and we need to find a single letter $x_k$ that distinguishes $s[i:n]$ and $t[j:n']$
	(that is, we reached $k=1$, and we need to find the last letter of the word $w$
	which distinguishes $s[1:n]$ and $t[1:n']$).
	This can be easily found in $O(n)$ by applying,
	e.g., a similar strategy as the one in the proof of the claim. 
	
	The overall time complexity is $O(n)$ as the time we use is: 
	\begin{align*}
		& O(n) \\
		+ & O(|\alp(a_{k-1})|+|\alp(b_{k-1})|\\
		& \ \  + (\nextpos(1, x_1)-1)+(\nextpos(1, x_1)-1))\\
		+ & O(|\alp(a_{k-2})|+|\alp(b_{k-2}|\\
		& \ \  + (\nextpos(\nextpos(1, x_1)+1, x_2)-\nextpos(1, x_1))\\
		& \ \  + (\nextpos(\nextpos(1, x_1)+1, x_2)-\nextpos(1, x_1)))\\
		+ & \ \ldots. 
	\end{align*}
	
	Adding this up, we get that the overall time we use is $O(|T_s|+|T_t|+n)=O(n)$.
\end{proof}

\section{Conclusions and future work}
\label{sec:conclusions}
In this paper,
we presented the first algorithm solving \maxkproblem\ in optimal time.
This algorithm is based on the definition and efficient construction of a novel data-structure:
the Simon-Tree associated to a word. 
Our algorithm constructs the respective Simon-Trees for the two input words of \maxkproblem,
and then establishes a connection between their nodes.
While the Simon-Tree is a representation of the classes induced,
for all $k\geq 1$,
by the $\sim_k$-congruences on the set of suffixes of a word,
this connection allows us to put together the classes induced by the respective congruences on the set of suffixes of both input word,
and to obtain, as a byproduct, the answer to \maxkproblem.\looseness=-1

The work presented in this paper can be continued naturally in several directions.
For instance, it seems interesting to us to compute efficiently,
for two words $s$ and $t$,
what is the largest $k$ such that $\SF_{\leq k} (s)\subseteq \SF (t)$. 
Similarly, one could consider the following pattern-matching problem:
given two words $s$ and $t$, and a number $k$,
compute efficiently all factors $t[i:j]$ of $t$ such that $t[i:j]\sim_k s$.
Finally, \kdecision\ could be extended to the following setting:
given a word $s$ and regular (or a context-free) language $L$, and a number $k$,
decide efficiently whether there exists a word $t\in L$
such that $s\sim_k t$.
A variant of \maxkproblem\ can be also considered in this setting:
given a word $s$ and regular (or a context-free) language $L$,
find the maximal $k$ for which there exists a word $t\in L$ such that $s\sim_k t$.\looseness=-1


\newpage
\bibliography{references}

\newpage
\appendix
\section{Computational model} \label{sec:CompModel}

In the standard Word Ram model of computation with words of size $\omega$, the memory consists
of memory-words consisting of $\omega$ bits each. 
Basic operations (including arithmetic and bitwise Boolean operations)
on memory-words take constant time,
and indirect addressing also takes constant time, that is given a memory-word storing an integer $a$ we can
access the memory-word $a$.
Denoting the size of the input by $n$, the input is given in the first $n$ memory-words,
and we assume $\omega=\Omega(\log n)$ as to allow for accessing all of its entries.
In some applications, $\omega$ might be much larger, see e.g.~\cite{BelazzouguiBN14},
however in our case it is natural to assume $\omega=\Theta(\log n)$.

In the Word RAM model with memory-word size $\Theta(\log n)$, the input to our problem
are two sequences of integers, each such integer stored in a single memory-word.
Under the assumption on the memory-word size, we can radix-sort a sequence in $O(n)$ time.
This allows us to replace every element by its rank in the sequence obtained
by sorting all elements occurring in the input and removing the duplicates.
Then, the alphabet can be assumed to be $\{1,\ldots,n\}$. A similar
renaming procedure can be applied for larger memory-word size $\omega$, as long we can
sort in linear time or allow randomisation, this is however not the focus
of this paper. 

It is not unusual to consider more restricted models of computation, in which
we forbid arbitrary arithmetic/bitwise manipulations on the input. As an extreme example,
on the one hand, we might be only allowed to test equality of the elements from the input, or, on the other hand, 
only check if one is smaller than the other. While there are nontrivial problems
that can be efficiently solved in the former model (the one allowing only equality testing), such as real-time pattern matching~\cite{BreslauerGM13},
testing if $s\sim_{1} t$ in this model requires $\Omega(n^{2})$ equality-tests.
To see this, observe that we can reduce checking equality of two sets $S$ and $T$
of size $\leq n$ to testing $s\sim_{1} t$ by writing down the elements in every set
as a sequence in an arbitrary order, and then apply a simple adversary-based lower bound
for the former. In the latter model (where general comparison-tests are allowed), known as the general ordered alphabet, the same
reduction shows that $\Omega(n\log n)$ comparison-tests are necessary~\cite{dobkin}. 
Accordingly, both \kdecision\ and \maxkproblem\ require $\Omega(n\log n)$ time in this model.

An interesting extension of our results is an efficient algorithm for the general ordered
alphabet, following previous papers on strings algorithms in this direction
(see \cite{dimaRuns,dimaCF,pawelLCE} and the references therein). 
An important point is that, while we restrict the operations on the input, the algorithm
is allowed to internally exploit, say, bit-parallelism. More formally,
the input is a sequence of elements from a totally ordered set ${\mathcal U}$
(i.e., string over ${\mathcal U}$).
The operations allowed in this model are those of the standard Word RAM model,
with one important restriction:
the elements of the input cannot be directly accessed nor stored in the memory used by the algorithms;
instead, we are only allowed to {\em compare}
(w.r.t. the order in ${\mathcal U}$)
any two elements of the input in constant time.
The lower bound of $\Omega(n\log n)$ comparisons by a reduction from set-equality still holds
in this model. To complement this, we first sort the input in $O(n\log n)$ time,
reduce the alphabet to $\{1,\ldots,n\}$, and then spend additional $O(n)$ time to solve
the problem. In this case, there is no need to use constant-time union-find and split-find data
structures, and a simple amortised logartihmic-time implementation suffices to achieve
$O(n\log n)$ time solution matching the lower bound.

\end{document}